\documentclass{article}

\usepackage{amsfonts,amsmath,amssymb}

\usepackage{enumerate}
\usepackage{times}
\usepackage{bm}
\usepackage{natbib}
\usepackage[percent]{overpic}
\usepackage{float}
\usepackage{graphicx}

\usepackage{subcaption}

\usepackage{lipsum}
\usepackage{bm}
\usepackage{natbib}
\usepackage[utf8]{inputenc}
\usepackage{amssymb}
\usepackage[justification=centering]{caption}
\usepackage{mathrsfs}

\usepackage{graphicx}
\usepackage{mathtools}
\usepackage{booktabs}
\usepackage{amstext} %
\usepackage{titling}
\usepackage{epstopdf}
\usepackage{array} 
\usepackage{booktabs}
\usepackage{caption}
\usepackage{epstopdf}
\usepackage{xcolor}
\usepackage{comment}

\DeclareMathOperator{\tr}{tr}
\DeclareMathOperator{\cov}{cov}
\DeclareMathOperator{\var}{var}
\DeclareMathOperator{\pr}{pr}
\DeclareMathOperator{\adj}{adj}
\DeclareMathOperator{\diag}{diag}

\newcommand{\calR}{\mathcal R}
\newcommand{\calS}{\mathcal S}

\newcommand{\calL}{\mathcal L}
\newcommand{\calC}{\mathcal C}
\newcommand{\calH}{\mathcal H}

\newcommand{\calF}{\mathcal F}

\newcommand{\bD}{ { D} }

\newcommand{\bg}{ { g} }

\newcommand{\bh}{ { h} }

\newcommand{\bI}{ { I} }

\newcommand{\bK}{ { K} }

\newcommand{\bX}{ { X} }

\newcommand{\bY}{ { Y} }

\newcommand{\boeta}{{\eta}}
\newcommand{\eps}{\epsilon}

\newcommand{\taus}{\tau^2}

\newcommand{\given}{\,|\,}

\newcommand{\bbeta}{ { \beta} }

\newcommand{\beps}{ { \epsilon} }

\newcommand{\bphi}{ { \phi} }
\newcommand{\bPhi}{ { \Phi} }

\newcommand{\bSigma}{ { \Sigma} }

\newcommand{\calU}{{\cal U}}

\usepackage{ulem}

\usepackage{xr}

\usepackage{amsfonts,amsmath,amssymb}

\usepackage{enumerate}
\usepackage{times}
\usepackage{bm}
\usepackage{natbib}
\usepackage[percent]{overpic}
\usepackage{float}
\usepackage{graphicx}

\usepackage{subcaption}

\usepackage{lipsum}
\usepackage{bm}
\usepackage{natbib}
\usepackage[utf8]{inputenc}
\usepackage{amssymb}
\usepackage[justification=centering]{caption}
\usepackage{mathrsfs}
\usepackage{graphicx}
\usepackage{mathtools}
\usepackage{booktabs}
\usepackage{amstext} %
\usepackage{epstopdf}
\usepackage{array} 
\usepackage{booktabs}
\usepackage{caption}
\usepackage{epstopdf}
\usepackage{xcolor}
\usepackage{comment}
\usepackage{authblk}

\usepackage{authblk}

\usepackage{ulem}

\usepackage{xr}
\usepackage{amsthm}
\newtheorem{proposition}{Proposition}
\newtheorem{theorem}{Theorem}
\newtheorem{lemma}{Lemma}

\usepackage[margin=1.25in]{geometry}

\usepackage{setspace}
\newcommand{\supplementstart}{
    \setcounter{section}{0} 
    \renewcommand{\thesection}{S\arabic{section}} 
    \title{Supplementary materials for ``\maintitle"}
    \maketitle
}

\begin{document}

\newcommand{\maintitle}{Consistency of common spatial estimators under spatial confounding}
\title{\maintitle}

\author[1]{Brian Gilbert}
\author[2]{Elizabeth L. Ogburn}
\author[2]{Abhirup Datta}

\affil[1]{Division of Biostatistics, NYU Grossman School of Medicine}
\affil[2]{Department of Biostatistics, Johns Hopkins University}

\maketitle

\begin{abstract}
   This paper addresses the asymptotic performance of popular spatial regression estimators of the linear effect of an exposure on an outcome under ``spatial confounding" --- the presence of an unmeasured spatially-structured variable influencing both the exposure and the outcome. We first show that the estimators from ordinary least squares (OLS) and restricted spatial regression are asymptotically biased under spatial confounding. We then prove a novel main result on the consistency of the generalized least squares (GLS) estimator using a Gaussian process (GP) working covariance matrix in the presence of spatial confounding under infill (fixed domain) asymptotics.
    The result holds under very general conditions --- for any exposure with some non-spatial variation (noise), for any spatially continuous fixed confounder function, using any Mat\'ern or square exponential kernel used to construct the GLS estimator, and without requiring Gaussianity of errors. Finally, we prove that spatial estimators from GLS, GP regression, and spline models that are consistent under confounding by a fixed function will also be consistent under endogeneity or confounding by a random function, i.e., a stochastic process. We conclude that, contrary to claims in some literature on spatial confounding, traditional spatial estimators are capable of estimating linear exposure effects under spatial confounding as long as there is some noise in the exposure. We support our theoretical arguments with simulation studies.
\end{abstract}

\section{Introduction}\label{sec:intro}

``Spatial confounding" commonly refers to the phenomenon where an unmeasured variable that is ``spatially structured'', i.e., a continuous or `smooth' function of space, influences both the outcome and the exposure variables. Understanding the impact of spatial confounding on popular estimators of the exposure effect in geospatial models has become a very popular topic in spatial and environmental statistics. %
However, the literature has offered conflicting conclusions and confused accounts of the origins of and solutions to spatial confounding, without much asymptotic theory to support the claims. In this paper, we offer consistency (or lack thereof) of popular estimators based on spatial regression models. We focus on the ordinary least squares (OLS) estimator, restricted spatial regression (RSR), and generalized least squares (GLS) estimator using Gaussian process (GP) covariance, but also discuss spline regression and Bayesian GP regression. Our main result is establishing consistency of GLS under very general forms of spatial confounding in fixed-domain (infill) asymptotics, as long as there is some noise (non-spatial variation) in the exposure process. %

One of the earliest discussions of spatial confounding is found in \cite{clayton}, who noted that introducing a spatial error term to a linear model %
can induce large changes in coefficient estimates relative to a non-spatial model. 
They recommend its inclusion anyway since unmeasured confounding is often a strong threat to study validity. \cite{paciorek} defines spatial confounding as the presence of an unmeasured spatially-varying confounder that is correlated with the exposure, derives the finite sample bias of the GLS estimator and studies it analytically and via extensive empirical studies, concluding that traditional spatial estimators like GLS and splines are favorable to OLS estimators in most realistic scenarios of spatial confounding.

In contrast, \cite{reich} and \cite{wakefield} %
view the shift in effect estimates due to adding a spatial error to a linear model as problematic. \cite{hodges} introduce restricted spatial regression (RSR) as a workaround, %
restricting the spatial random function to be geometrically orthogonal to covariates. %
While the authors acknowledge that the existence of an unmeasured spatial confounder implies bias in this estimate, they claim that traditional spatial models cannot adjust for such bias either. Recent papers have built off RSR while maintaining its essential features \citep{hughes, prates, azevedo}. RSR has come under criticism recently. \cite{hanks,khan,zimmerman} found that %
uncertainty quantification from RSR can be severely anticonservative. %
While this line of work makes a convincing case for not using RSR due to its second-order properties (variance), it does not consider data generation processes with spatial confounding and thus does not address the original claim that RSR is likely no more biased than traditional spatial regression models under spatial confounding. 

Some studies have %
considered data generation processes (DGP) with explicit spatial confounding to derive analytical expressions of the biases and mean squared errors of the spatial estimators such as GLS \citep{paciorek,page,schnell,nobre2021effects,khan2,narcisi2024effect}. These finite-sample bias expressions have sometimes led to claims that the GLS estimator is biased. For example, \cite{page} notes the equivalence of spatial confounding and endogeneity and argues that the GLS estimator is biased under endogeneity. \cite{schnell} points to the finite sample bias  of the GLS estimator and claims that common spatial estimators do not mitigate bias under spatial confounding. However, in the presence of unmeasured spatial confounding, no estimator of the exposure effect will be unbiased in finite samples unless making very strong assumptions on the confounder. An important factor in the evaluation of an estimator is asymptotic consistency, which complements analytical or empirical studies of the finite sample bias under various data generation scenarios and sample sizes. In this regard, \cite{khan2} provides a comprehensive study of non-spatial or unadjusted (OLS) and spatial (GLS) estimators under data generation processes with explicit spatial confounding, expanding on the previous work of \cite{paciorek}. They derive analytic expressions for the finite-sample bias of OLS and GLS estimators 
and show that, under most scenarios of spatial confounding, the spatial estimators like GLS will generally reduce bias compared to the OLS estimator. 
\cite{khan2} also offers numerical results from a suite of different data generation scenarios which provide additional empirical evidence in favor of using the GLS estimator under spatial confounding.  
 However, to our knowledge, none of these papers studied how the GLS or OLS estimators will behave asymptotically under spatial confounding. This is the knowledge gap that we fill here.

For our study, we also make a clear distinction between the data generation process (DGP) and the analysis models used to derive estimators of the exposure effect. The DGP for the outcome is specified as the widely studied and used spatial \textit{partially linear model (PLM)}. %
We consider two scenarios: one where the unmeasured confounder is a fixed spatial function and one where it is a random spatial function, i.e., a stochastic process on the spatial domain.  
We show how spatial confounding can be defined from first principles of causal inference and establish conditions for identifiability of the linear exposure effect for data generated from the PLM.  %
We then present a direct first-order argument against estimators like the unadjusted OLS estimator and the RSR estimator by quantifying their non-vanishing asymptotic error under %
spatial confounding. %

Our main result then establishes the consistency of the GLS estimator under spatial confounding as long as there is some non-spatial variation in the exposure. The result holds under very general conditions -- any random sampling scheme of locations in a fixed domain (infill asymptotics), spatial confounding by any continuous fixed function of space, errors for both the outcome and exposure deviating from normality, and the GP covariance function used in GLS being any universal kernel (e.g., Mat\'ern or square exponential) with any choice of kernel parameters. %
The proof uses Mercer representation of the reproducing kernel Hilbert space (RKHS) of Gaussian processes with a universal kernel %
to show that the GLS procedure essentially prewhitens the outcome and the exposure, removing the spatial confounder and allowing recovery of the exposure effect from the non-spatial part of the exposure. Similar ideas of prewhitening have been more explicitly exploited in joint outcome and exposure models to mitigate spatial confounding \citep{thaden,dupont}. We show that the traditional GLS estimator achieves the same goal without explicitly specifying the exposure model. The result thus dispels claims (in an asymptotic sense) that GLS does not mitigate bias under spatial confounding \citep[see, e.g.,][]{hodges,page,schnell}. Our result offers asymptotic confirmation to the empirical and analytical finite-sample conclusions of \cite{paciorek} and \cite{khan2} that in the presence of some non-spatial variation in the exposure, traditional spatial estimators can be used even in the presence of spatial confounding.

Last, we consider spatial confounding by a random function %
of space. This is a form of ``endogeneity,” as omitting the spatial random function (i.e., a random effect) would result in endogenous error terms that are correlated with exposure. %
We show that GLS is also consistent under very general conditions of endogeneity ---   
not requiring any distributional assumptions (e.g., stationarity or Gaussianity) on the confounding stochastic process beyond having continuous sample paths. 
We offer a general result on when the consistency of an estimator under spatial confounding by a fixed function will imply consistency under endogeneity. %
This result is broadly applicable to other estimators like splines and Gaussian processes (GP), 
showing that existing results on their consistency for the PLM under spatial confounding by a fixed function, e.g., \cite{rice,dupont} for splines and  %
\cite{yang} for GP, also imply consistency of these estimators under spatial endogeneity. This result contradicts the perception that estimators from exogenous analysis models like the GLS or GP regression estimator, where error processes of the outcome are modeled as independent of the exposure, cannot account for endogeneity \citep[see e.g.,][]{bell}. %
This seemingly paradoxical result on the consistency of exogenous estimators under endogeneity is due to the spatially smooth nature of the confounder, demonstrating the benefit of spatial confounding as opposed to an unstructured or clustered confounder.
We corroborate the theoretical findings through synthetic experiments. 

\section{Data generation process}\label{sec:plm}
\subsection{Outcome model}\label{sec:dgpy}
Throughout, we will assume the following data generation process (DGP). Let $\calR \subset \mathbb R^d$ be a %
convex and compact
spatial domain. The locations $S_1,\ldots,S_n$ are sampled randomly (i.i.d) from a sampling density $f_s$ with support on the whole of $\calR$. Our asymptotic results will keep the domain $\calR$ fixed, i.e., we consider infill asymptotics with random sampling of an increasing number of locations. For $s \in \calR$, let $Y(s)$ be a univariate outcome and $X(s)=(X_1(s),\ldots,X_p(s))^\top$ be a $p$-dimensional exposure observed at $s$. %
Let $X(\cdot)$ denote the process $\{X(s)\}_{s \in \calS}$ while $X$ denotes the matrix (or vector) formed by vertically stacking $X(S_1)^\top, X(S_2)^\top,\ldots, X(S_n)^\top$. We use the same convention notation for all other fixed functions or random functions  (stochastic processes) on $\calR$, e.g., $Y(\cdot)=\{Y(s)\}_{s \in \calR}$ and $Y=(Y(S_1),\ldots, Y(S_n))^\top$. 
The DGP for the outcome process is given by a \textit{partially linear model (PLM)}  %
\begin{equation}\label{eq:plm}
Y(s) = \bX(s)^\top \bbeta^* + g(s) + \eps(s), \eps(s) \sim_{i.i.d}, E(\eps(s)) = 0, \var(\eps(s)) < \infty, \eps(\cdot) \perp X(\cdot). 
\end{equation}

This PLM features extensively in spatial analysis and study of spatial confounding, both as the data generation model and as the analysis model. The spatial function $g(\cdot): \calR \to \mathbb R$ %
can either be a fixed function or a random function; %
the empirical properties of a single draw cannot determine which is the case. Some of the literature on spatial confounding decides on the nature of the true $g(\cdot)$ based on the analysis model used. For example, if $g(\cdot)$ is modeled using splines or basis functions and estimated using penalized regression, the true $g(\cdot)$ is also assumed to be fixed \citep{rice,dupont}. When $g(\cdot)$ is modeled as a random function, typically as a Gaussian process (GP), the true function $g(\cdot)$ is also assumed to be a random function \citep{hodges,khan,zimmerman}. This latter practice can be problematic: treating the true $g(\cdot)$ as a GP, which is traditionally modeled to be independent of the exposure $X(\cdot)$, ignores their possible true correlation in the DGP and led to overlooking the possibility of bias of certain estimators under spatial confounding. We correct these misapprehensions in Section \ref{sec:bias}.

It is thus important to disentangle the DGP from the analysis model. %
Regardless of the analysis model used, we view it as preferable to treat the true $g(\cdot)$ in the DGP as a \textit{fixed} function of space as we are typically interested in inference about a population that would share the realization of $g(\cdot)$ that generated the data. Note that assuming $g(\cdot)$ to be a fixed function in the DGP is perfectly compatible with an analysis model that models $g(\cdot)$ as a random function like a Gaussian process or a spline or basis function expansion with priors for the coefficients; this is analogous to a Bayesian model that treats a fixed parameter as random during the analysis. %
Occasionally, there can be compelling reasons to consider a random function $g(\cdot)$ in the DGP, such as a desire to generalize findings to future, independent generations of the data.  %
In such cases, studying the bias of an estimator requires explicit consideration of the correlation between this random function and the exposure (i.e., endogeneity), as done in \cite{paciorek,page,schnell,nobre2021effects,khan2,narcisi2024effect}. Ignoring this correlation results in erroneous conclusions of unbiasedness. We consider the DGP with confounding by a random spatial function in Section \ref{sec:random}.

\subsection{Spatial confounding}\label{sec:conf}

Spatial confounding occurs in the PLM (\ref{eq:plm}) when the exposure function $X(\cdot)$ is correlated with the fixed spatial function $g(\cdot)$ in the outcome, i.e., 
\begin{equation}\label{eq:confcor}
    \cov(g(S),X(S)) \neq 0, \mbox{ where } S \sim f_s. 
\end{equation}
Condition (\ref{eq:confcor}), established in Section \ref{sec:confder} of the Supplement using potential outcomes, is meaningful even without the causal inference formalism used to derive it. If $g(\cdot)$ is a fixed function, then the correlation in (\ref{eq:confcor}) accounts for randomness in $g(S)$ from the random sampling of the location $S$, i.e., $\cov(g(S),X(S)) = \int_\calR E(X(s)) g(s) f_s(s) ds - \int_\calR E(X(s)) f_s(s) ds \int_\calR g(s) f_s(s) ds$. If $g(\cdot)$ is a random function then randomness in $g(S)$ comes from both $g(\cdot)$ and $S$ (See Section \ref{sec:confder} for the expression). 
As $X(\cdot)$ is $p$-dimensional, for $p > 1$ Equation (\ref{eq:confcor}) is a vector equation, i.e., there is a nonzero correlation between $g(\cdot)$ and at least one component of $X(\cdot)$. 

\subsection{Identifiability and exposure model}\label{sec:uniq}
 
In \eqref{eq:plm}, since both $g(\cdot)$ and (possibly) the exposure $X(\cdot)$ vary with space, additional assumptions are required on the exposure process to ensure that the parameters in \eqref{eq:plm} are identified, in the sense that only one pair $(\beta^*, g(\cdot))$ is compatible with infinite data generated from \eqref{eq:plm}.  
For RSR, \cite{hodges} restricted $g(\cdot)$ to the orthogonal complement ($X(\cdot)^\perp$) of $X(\cdot)$, %
ensuring identifiability. However, as we see from (\ref{eq:confcor}), 
assuming $g(\cdot) \in X(\cdot)^\perp$ implies a DGP without any spatial confounding. %
Below, we provide a necessary and sufficient criterion for identifiability in the PLM that accommodates confounding. %
All proofs are in the Supplement. \\

\begin{proposition}\label{prop:plm_id}
Let $\mathcal{F}$ be a class of functions from $\calR$ to $\mathbb R$ that is closed under pairwise addition and scalar multiplication. Let $x(\cdot)=(x_1(\cdot),\ldots,x_p(\cdot))^\top$ %
denote a draw (realization) of $X(\cdot)$.
Then, given $X(\cdot)=x(\cdot)$, the partially linear model (\ref{eq:plm}) has a unique solution in $\mathbb R^p \times \calF$, in the sense that there is only one pair $(\beta^*,g(\cdot))  \in \mathbb R^p \times \calF$ for which the data have the distribution given in equation $(\ref{eq:plm})$ for all $s \in 
\calR$, if and only if  %
$x(\cdot)^\top a \notin \mathcal{F}$ for all nonzero $a \in \mathbb R^p$.\\ %
\end{proposition}

The PLM is thus identified when the confounder $g(\cdot)$ belongs to the class $\calF$, and the exposure $X(\cdot)$ (or, in the multivariate case, any linear combination of its components) does not belong to $\calF$. As function classes are typically defined by smoothness (e.g., Hölder smoothness) and smoother classes are nested inside less smooth ones, this implies that the exposure needs to have a non-smooth component relative to $\calF$. The identifiability result thus agrees with earlier literature suggesting the need for the exposure $X(\cdot)$ to vary on a smaller spatial scale than the confounder $g(\cdot)$ to identify the exposure effect. We note that \textit{spatial scale} has been defined and interpreted in different ways in the literature and is often tied to specific models. \cite{paciorek} and \cite{khan2} define it via the range parameters of the Gaussian processes assumed to model the exposure and the confounder. \cite{schnell} defines it through the parameters of the conditional autoregressive models for $X(\cdot)$ and $g(\cdot)$. \cite{guan} and \cite{keller} define spatial scales in terms of the coefficients of Fourier or other basis expansions of the two functions. More generally, \cite{schnell} and \cite{gilbert} discuss how differences in the spatial scales of the exposure and confounder are related to the \textit{positivity} assumption in causal inference. Proposition \ref{prop:plm_id} does not use such specific distributional assumptions and defines spatial scale through the smoothness of the function class $\calF$ for the confounder $g(\cdot)$. If $X(\cdot) \notin \calF$, then it varies at finer scales than $g(\cdot)$. 
Note that for Proposition \ref{prop:plm_id}, the exposure $X(\cdot)$ need not be restricted to the orthogonal complement of $\calF$, as in RSR, but can be correlated with $g(\cdot)$, thereby accommodating spatial confounding. 

Proposition \ref{prop:plm_id} motivates considering the following data generation process for the exposure:
\begin{equation}\label{eq:exp}
    X(s) = h(s) + \eta(s),\, \eta(s) \sim_{i.i.d}, E(\eta(s)) = 0, O \prec \var(\eta(s)) \prec \infty, \eta(\cdot) \perp \eps(\cdot), 
\end{equation}
where $h(\cdot)=(h_1(\cdot),\ldots,h_p(\cdot))^\top$, each $h_j(\cdot)$ being some fixed function of space. 
\cite{yang} assumes a similar DGP as (\ref{eq:exp}) for the exposure in their study of Bayesian Gaussian process estimators for the PLM (\ref{eq:plm}). The DGP is quite general, allowing $X(\cdot)$ to have both a spatial component $h(\cdot)$ correlation of which with $g(\cdot)$ effectuates spatial confounding, and a non-spatial noise component $\eta(s)=(\eta_1(s),\ldots,\eta_p(s))^\top$ with zero mean (a nonzero mean can be absorbed into $h(\cdot)$ as a constant function of space, i.e., an intercept) and a strictly positive-definite variance. The latter precludes the scenario that any linear combination $\eta(s)^\top a$ is degenerate for any non-zero $a \in \mathbb R^p$, ensuring that sample paths (draws) of the i.i.d error process $\eta(\cdot)^\top a$ are almost surely of zero smoothness. This, in turn, guarantees that no linear combination $X(\cdot)^\top a$ belongs to any smooth function class $\calF$ considered for $g(\cdot)$, implying identifiability from Proposition \ref{prop:plm_id}.

\section{Analysis models and the consistency of their estimators}\label{sec:consist}
\subsection{Analysis models}\label{sec:anmodel}
We consider various \textit {analysis models} for estimating the effect of exposure $X(s)$ on the outcome $Y(s)$ in the PLM (\ref{eq:plm}). Akin to \cite{khan2}, we distinguish between the data-generation process (DGP), specified in (\ref{eq:plm}) and (\ref{eq:exp}), and the analysis models used because, crucially, they need not coincide: in some cases an analysis model may result in consistent estimation \textit{even if} it makes modeling assumptions that conflict with the DGP. Recognizing this distinction is crucial to unraveling the confusion around spatial confounding.
 The OLS estimator of $\beta^*$ is 
     $\hat\beta_{OLS} = (  X^\top  X)^{-1} X^\top Y$. 
 The OLS estimator arises from an analysis model that assumes that $g(\cdot)$ is absent in \eqref{eq:plm}. It is thus often referred to as the \textit{unadjusted} estimator, as it does not attempt to adjust for spatial confounding. Common estimators that adjust for a spatial component in the outcome model include the estimator of $\beta$ from a \textit{spline regression} \citep{rice} which models $g(\cdot)$ as a fixed function of space, represented using splines. Alternatively, $g(\cdot)$ is also commonly modeled as a draw from a random function, i.e., a spatial stochastic process  %
$\nu(\cdot)$ endowed with a Gaussian process (GP) prior with zero mean and some covariance function $K(\cdot,\cdot)$. %
The GP regression analysis model can be summarized as a hierarchical spatial mixed effects linear model:
\begin{equation}\label{eq:gp}
    Y(s) =  {X(s)^\top} \beta + \nu(s) + \epsilon(s), \nu(\cdot) \sim GP(0,K(\cdot,\cdot)), \epsilon(s) \sim_{i.i.d} N(0,\taus).
\end{equation} %
If $\beta$ is also assigned a prior, then Bayesian (MCMC) algorithms are typically used to estimate $\beta$ and $\nu$. %
Alternatively, the GLS estimator  arises as the maximum likelihood estimator of $\beta$ from the marginal model $ Y = N(X \beta, \Sigma)$ obtained from integrating over the GP prior of $\nu$ in (\ref{eq:gp}), i.e., 
\begin{equation}\label{eq:gls}
    \hat\beta_{GLS} = (X^\top\Sigma^{-1}X)X^\top\Sigma^{-1}Y,\,  \Sigma = K + \tau^2 I \mbox{ with } K_{ij} = K(S_i,S_j).
\end{equation}
We reemphasize that (\ref{eq:gp}) is an analysis model here used to derive the GLS estimator. It is not suitable for being a DGP as the random function $\nu(\cdot)$, modeled to be independent of the exposure $X(\cdot)$, precludes any mechanism of spatial confounding.
 The RSR analysis model is similar to the mixed effects model (\ref{eq:gp}) with the additional restriction that %
 $\nu(\cdot)$ is geometrically orthogonal to $X(\cdot)$, or in practice, $\nu^\top X = 0$ where $\nu=(\nu(S_1),\ldots,\nu(S_n))^\top$. This geometric orthogonality is different from prior specification of statistical independence between $\nu(\cdot)$ and $X(\cdot)$ specified in (\ref{eq:gp}) which does not guarantee geometric orthogonality between sample paths of $X(\cdot)$ and the posterior estimate of $\nu(\cdot)$. The hard constraint of explicit orthogonality imposed by RSR results in the same maximum likelihood estimator for $\beta^*$ as OLS, i.e., $\hat \beta_{RSR} = \hat \beta_{OLS}$. 

\subsection{Bias of OLS and RSR under spatial confounding}\label{sec:bias}

We quantify the asymptotic bias of the unadjusted OLS  estimator ${\hat \beta_{OLS}}$ and equivalently, the RSR estimator (as $\hat \beta_{OLS}=\hat \beta_{RSR}$) under spatial confounding.\\ 
\begin{proposition}\label{prop:ols}
    Consider locations $S_1,\ldots,S_n \sim_{i.i.d} f_s$, %
    $Y(S_i)$ generated from (\ref{eq:plm}) with a continuous $g(\cdot)$, $X(S_i)$ are generated from (\ref{eq:exp}) with a continuous $h(\cdot)$ for $i=1,\ldots,n$ such that there is spatial confounding as defined in (\ref{eq:confcor}). %
    Then, the unadjusted OLS estimator or RSR estimator $\hat \beta_{OLS}$ satisfies $\hat\beta_{OLS} - \beta^* \to_p \var(X(S))^{-1}\cov(X(S), g(S)) \neq 0$, where $S \sim f_s$.\\ %
\end{proposition}

Proposition \ref{prop:ols} quantifies the bias of the unadjusted OLS or RSR estimator under very general conditions of spatial confounding (requiring no assumptions on the functions $g(\cdot)$ and $h(\cdot)$ beyond continuity and satisfying the spatial confounding condition, and no assumptions on the error processes of the exposure and outcome beyond finite second-order moments.

We give the intuition of the proof here. If $g(\cdot)$ were a \textit{known} function, then the adjusted OLS regression of $Y$ onto $X$ with the offset $g =(g(S_1),\ldots,g(S_n))^\top$ would yield an unbiased  estimator of $\beta^*$. We call this \textit{oracle} estimator $\hat  \beta_{oracle} = (X^\top  X)^{-1}X^\top (Y-g) $. Then 
\begin{equation}\label{eq:biasols}
    \hat \beta_{OLS} - \hat  \beta_{oracle} =    (  X^\top  X)^{-1}   X^\top Y -(  X^\top  X)^{-1}   X^\top (Y-g) 
= (  X^\top   X)^{-1}   X^\top g.
\end{equation}
Since $\hat{\beta}_{oracle}$ is unbiased and consistent (as it is an OLS estimate from a correctly specified model), we can immediately see that the non-spatial model is biased under spatial confounding (that is, if $X(\cdot)$ and $g(\cdot)$ are correlated), 
with error $ (X^\top X)^{-1} X^\top  g$. The asymptotic error given in Proposition (\ref{prop:ols}) is the population limit of this error and is nonzero if there is spatial confounding. 
 
 This first-order bias of the OLS or RSR estimator is unsurprising and stems from the well-known issue of omitting a confounding variable in a regression. Recent criticisms of the RSR have mostly focused on its second-order properties, showing that variance of the RSR estimator is anti-conservative \citep{khan,zimmerman}. These studies assumed the mixed effects analysis model (\ref{eq:gp}) as the true DGP, in which case there is no spatial confounding (as $\nu(\cdot)$ is assumed to be independent of $X(\cdot)$), and thus RSR is unbiased. %
 We provide a direct first-order bias result for OLS or RSR under any scenario of spatial confounding. This highlights the importance of separating the DGP from the analysis model.

\subsection{Consistency of GLS under spatial confounding}%
\label{sec:fixed}

The omitted variable bias of the OLS or RSR estimator under spatial confounding is unsurprising. Similar logic has been used to argue that a spatial GLS estimator (\ref{eq:gls}) cannot adjust for confounding either \citep[e.g.,][]{hodges}. This notion likely arises from the fact that the GLS estimator is derived from the misspecified model (\ref{eq:gp}): it marginalizes over a spatial random function $\nu(\cdot)$ with a GP prior assumed to be independent of the exposure $X(\cdot)$, completely ignoring their correlation in the DGP. We now establish the main result --- consistency of the GLS estimator under spatial confounding, contrary to this common belief and intuition. \\

\begin{theorem}\label{thm:gls} Consider locations $S_1,\ldots,S_n \sim_{i.i.d} f_s$,
    $Y(S_i)$ generated from (\ref{eq:plm}) with a continuous $g(\cdot)$, with $X(S_i)$ generated from (\ref{eq:exp}) with a continuous $h(\cdot)$ for $i=1,\ldots,n$. 
    Consider the GLS estimate $\hat\beta_{GLS}$ in (\ref{eq:gls}) using a working covariance matrix $\bSigma=K + \tau^2 I$ 
 where $0 < \tau^2 < \infty$ and $K(\cdot,\cdot)$ is a stationary Matérn kernel with any fixed set of parameters (variance $0 < \sigma^2 < \infty$, spatial decay $0 < \phi < \infty$, and smoothness $0 < \nu \leq \infty$). 
 Then $\hat\beta_{GLS} \to_p \beta^*$. \\
\end{theorem}

Theorem \ref{thm:gls} proves the consistency of the spatial GLS estimator under spatial confounding using very mild assumptions. There is no assumption on the smoothness of $g(\cdot)$ and $h(\cdot)$ beyond continuity. There is also no restriction on the extent of spatial confounding, which comes from the possible correlation of $g(\cdot)$ with $h(\cdot)$, the spatial part of $X(\cdot)$. For example, extreme cases of confounding with some or all components $h_j(\cdot)$ being exactly equal to $g(\cdot)$ are accommodated. The only distributional assumptions on the error terms $\eps(s)$ and $\eta(s)$ are zero mean and finite (and, for $\eta(s)$, non-zero) variance. There is no assumption on the shape of these error distributions (e.g., normality), even though the GLS estimator is derived by marginalizing the mixed effects analysis model (\ref{eq:gp}) with normal errors. The GP kernel $K(\cdot,\cdot)$ used to construct the working covariance matrix $\Sigma$ can be any Mat\'ern or square exponential kernel (for $\nu=\infty$) without any restriction on the parameter values used. 

Two things are central for the consistency result to hold. First, the nugget $\tau^2$ used in the working covariance matrix $\Sigma$ %
(which can be different from the true error variance $\tau_0^2=\var(\eps(s))$ needs to be strictly positive (even if there is no true nugget, i.e., $\tau^2_0=0$). Including a positive nugget in the working covariance matrix ensures that the eigen-values of $\Sigma^{-1}$ are bounded and $\hat \beta_{GLS}$ is well-behaved. Note that this is not an assumption on the DGP as the user can always force the  nugget in the working covariance matrix to be non-zero.
 
Second, the key piece to the GLS result is that the exposure must have some non-spatial variation $\eta(s)$, i.e., the $\var(\eta(s))$ must be strictly positive-definite, as specified in (\ref{eq:exp}). A similar assumption of additive noise in the exposure has been utilized in \cite{thaden} and \cite{dupont}, who also required this noise to be Gaussian, which we do not do here. Additive noise in the exposure ensures identifiability (see the discussion in Section \ref{sec:uniq}). We now provide an 
intuition as to why the GLS is consistent when this holds. 
The GLS estimator (\ref{eq:gls}) %
can be viewed as an OLS estimator based on  $Y^*=\Sigma^{-1/2}Y$ and $X^*=\Sigma^{-1/2}X$. From the DGP (\ref{eq:plm}) and (\ref{eq:exp}), we have  $X^*= \Sigma^{-1/2}h + \Sigma^{-1/2}\eta$, and $Y^*= \Sigma^{-1/2}h\beta^* + \Sigma^{-1/2}\eta\beta^* + \Sigma^{-1/2}g + \Sigma^{-1/2}\eps$. Crudely, a continuous function $g(\cdot)$ can be expanded approximately in terms of  eigenfunctions %
of a Mat\'ern kernel $K(\cdot,\cdot)$ (which is an \textit{universal kernel}) with larger weights assigned to the `flatter' or `low frequency' eigenfunctions (corresponding to the larger eigenvalues of $K(\cdot,\cdot)$ and thus of $\Sigma$). The term $\|\Sigma^{-1/2}g\|_2$ 
is approximately the norm of the weights scaled by the square root of the respective eigenvalues and will be small as the larger weights, corresponding to larger eigenvalues, will be scaled down more.  
A similar argument holds for $\|\Sigma^{-1/2}h\|_2$. Relatively, $\|\Sigma^{-1/2}\eta\|_2$ %
will be much larger as realizations of the noise process will be extremely non-smooth with high probability, giving large weights to even `high-frequency' eigenfunctions with smaller eigenvalues. 
So %
$\Sigma^{-1/2}g$ and $\Sigma^{-1/2}h$, involving the spatially smooth functions, are small in magnitude compared to $\Sigma^{-1/2}\eta$ involving the noise in the exposure. We then have
\[X^* \approx \Sigma^{-1/2} \eta \mbox{ and } E(Y^* \given X^*) = \Sigma^{-1/2}h\beta^* + \Sigma^{-1/2} \eta \beta^* + \Sigma^{-1/2} g \beta^* \approx \Sigma^{-1/2} \eta \beta^* \approx X^*\beta^*.\]
As %
$E(Y^* \given X^*)\approx X^*\beta^*$, the GLS estimator between $Y$ and $X$, which is the OLS estimator between $Y^*$ and $X^*$, is consistent for  $\beta^*$.  
The linear transformation using $\Sigma^{-1/2}$ thus approximately results in a residualization or \textit{prewhitening} of $Y$ and $X$, removing their spatial components. Regressing the residual part of the outcome on the residual noise component of the exposure is enough to identify $\beta^*$ as long as there is a nonzero noise (non-spatial) component in the exposure (ensured by assuming $var(\eta(s)) \succ O$). Residualization has been explicitly used in \cite{thaden} and \cite{dupont} to develop two-stage models that can identify the exposure effect under spatial confounding. Our result shows that the GLS carries out the residualization without  needing a two-stage analysis model. The formal proof of the GLS consistency result, provided in Section \ref{sec:gls_consist} of the Supplement, is based on the Mercer representation of the reproducing kernel Hilbert space (RKHS) of a GP with a universal kernel $K(\cdot,\cdot)$.

The inclusions of noise in the exposure generation and of a positive nugget in the working covariance matrix are essential to the consistency result. \cite{bolin2024spatial} has shown that the exposure effect cannot be consistently identified under spatial confounding if the outcome depends on a smoothed version of the exposure. Even if there is no confounding, \cite{wang2020prediction} shows that the GLS estimator is not consistent if the exposure is smooth and there is no nugget in $\Sigma$. Thus our result on consistency does not contradict these inconsistency results, which consider smooth exposures and/or a lack of a nugget in the working covariance matrix.

 The consistency of GLS under spatial confounding may seem counterintuitive, as reflected in the literature reviewed in Section \ref{sec:intro}. This is because the GP regression model (\ref{eq:gp}) from which GLS is derived
 is misspecified, assigning a GP prior for the confounder that is independent of the exposure, thereby apparently ignoring 
the correlation between the two. Indeed, the finite sample bias of the GLS estimator, as derived in \cite{paciorek,page,schnell,nobre2021effects,khan2}, stems from this misspecification. The bias is $E(t_{gX})$ where $t_{gX}=(X^\top\Sigma^{-1}X)^{-1}X^\top\Sigma^{-1}g$ whose expectation, conditional on $X$, will generally be non-zero if $g(\cdot)$ and $X(\cdot)$ are correlated. However, the asymptotic limit of $t_{gX}$ has not been studied previously and Theorem \ref{thm:gls} shows that this term vanishes asymptotically even under spatial confounding, leading to consistency of the GLS estmator and dispelling the myth about its unusability under spatial confounding. Our GLS consistency result aligns with the consistency of the estimator for $
\beta^*$ derived from Bayesian implementation of the GP regression (\ref{eq:gp}) \citep{yang}. The consistency of GLS for a fixed function $g(\cdot)$ has been
also shown previously by \cite{he}, but that work did not connect to the literature on spatial confounding,
and the result relies on technical assumptions that are hard to verify or interpret. Our result is also in accordance with the empirical and finite-sample analytical studies of \cite{paciorek} and \cite{khan2}, which demonstrated the reduction in bias from GLS estimators when the exposure varied at finer spatial scales than the confounder. The bias of OLS and the consistency of the GLS estimator are confirmed via numerical experiments in Supplementary Materials Section \ref{sec:sim_bias}. An approximate expression of the finite sample bias of the GLS estimator based on the RKHS representation is derived and empirically verified in Section \ref{sec:eigen}. The numerical experiments also reveal that interval estimates for GLS can offer poor coverage. This is not surprising as the interval estimates are derived from the analysis model (\ref{eq:gp}), which is clearly misspecified under confounding, and the estimates may be affected by so-called \textit{`regularization bias'}, as discussed in Section \ref{sec:sim_bias}. So while Theorem \ref{thm:gls} shows that the GLS point estimate is robust to the model misspecification, this may not be true of the second-order properties of the estimator.

\subsection{Consistency of spatial estimators under endogeneity}%
\label{sec:random}

We now consider a DGP where the functions $g(\cdot)$ and $h(\cdot)$ in (\ref{eq:plm}) and (\ref{eq:exp}) are %
random functions. %
Often $g(\cdot)$ and $\eps(\cdot)$ are grouped together to denote the `total error' process for the outcome. If $g(\cdot)$ is correlated with $h(\cdot)$, then $g(\cdot)+ \eps(\cdot)$  is correlated with the exposure $X(\cdot)$ leading to 
endogeneity. We now show that the GLS estimator is also consistent under such endogeneity.\\

 \begin{theorem}\label{thm:endo} Consider the DGP (\ref{eq:plm}) and (\ref{eq:exp}) where %
$(g(\cdot),h(\cdot))$ %
is a $(p+1)$-dimensional random function (stochastic process) on $\calR$ %
with almost surely continuous sample paths and independent of $\eps(\cdot), \eta(\cdot)$, and $S_1, S_2,\ldots$. Then for the same choice of working covariance matrix $\Sigma$ as in Theorem \ref{thm:gls}, $\hat \beta_{GLS} \to \beta^*$.\\
\end{theorem} 

The result is very general, imposing no distributional assumptions on the $(p+1)$-variate random function (stochastic process) $(g(\cdot),h(\cdot))$ beyond having continuous sample paths. There is also no restriction on the nature of the correlation between $g(\cdot)$ and $h(\cdot)$, thereby allowing many different scenarios of spatial confounding.
For example, $(g(\cdot),h(\cdot))$ can be any $(p+1)$-dimensional Gaussian process with any multivariate covariance kernel that leads to continuous sample paths. This includes all multivariate Mat\'ern GPs, including the extreme case of confounding where $g(\cdot)$ is a univariate Mat\'ern GP and $h_j(\cdot) = g(\cdot)$ for some or all $j$. Theorem \ref{thm:endo} also accommodates stochastic processes that are non-Gaussian and non-stationary, demonstrating the robustness of the consistency result for GLS to a wide variety of data generation mechanisms. Also note that for both Theorems \ref{thm:gls} and \ref{thm:endo}, the choice of the working covariance matrix $\Sigma$ is not tied to the true data generation process. Thus $\Sigma$ need not equal $\cov(Y)$ or $\cov(Y \given X)$. The only requirement for $\Sigma$ is to be of the form $K + \tau^2 I$ for a universal kernel $K(\cdot,\cdot)$ with any valid choice of parameters and a positive nugget $\tau^2$. In Section \ref{sec:simconsist} we empirically confirm the robustness of the GLS consistency result under confounding to covariance misspecification.

This result extending the consistency of GLS from fixed to random spatial functions in the DGP relies on a general result (Proposition \ref{prop:cond_consist} in the Supplement) about when consistency conditional on a given realization of a random function, i.e., a fixed function, implies marginal consistency accounting for the distribution of the random function. This proposition also applies to other estimators like splines or GP regression whose consistency for the PLM under spatial confounding has been established for fixed $g(\cdot)$ \citep{rice, yang}, proving that they will also be consistent, under respective assumptions, if $g(\cdot)$ is random.
Thus, when confounding is by a random spatial function, these estimators based on exogenous analysis models, that ignore the correlation between the exposure and the outcome errors, can control for endogeneity. This result, seemingly at odds with the conventional wisdom about endogeneity, holds due to the smoothness of the unmeasured confounder as opposed to a non-smooth or discrete unmeasured confounder and demonstrates the benefits of a spatially continuous confounding as opposed to unstructured or grouped confounding. We confirm this via simulations in Section \ref{sec:simcluster}.

\section{Conclusion}\label{sec:concl}
We have proved that in the presence of spatial confounding, the traditional GLS estimator based on a Gaussian process working covariance matrix can produce consistent effect estimates, despite not explicitly modeling the correlation between the confounder and exposure. The result holds as long as there is some non-spatial variation (noise) in the exposure, regardless of whether the confounder is a fixed or random function of space. In the latter case, the result holds without requiring any assumptions on the random functions beyond having continuous sample paths, thereby accommodating arbitrarily strong confounding (very strong correlation between $g(\cdot)$ and $X(\cdot)$), non-Gaussianity, and non-stationarity. %
The results also allow non-Gaussian errors in both the outcome and exposure and allows any universal kernel (e.g., Mat\'ern or square exponential covariance) with any choice of parameters as the GLS working covariance matrix. 

We also provide general results in the Supplementary section that show that notions of confounding by fixed and random functions are largely equivalent, which implies consistency of many other estimators (like splines and GP regression) under spatial endogeneity. On the other hand, the RSR or unadjusted OLS yield biased estimators under spatial confounding. 

Our study, throughout, separates the data generation process and the analysis model used to derive an estimate of the exposure effect. We define spatial confounding from the principles of causal inference via potential outcomes. We emphasize that the presence or absence of spatial confounding is determined by the data generation process, whereas the analysis model plays a role in determining whether an estimator derived from it can adjust for unmeasured spatial confounding. An unadjusted OLS estimator derived from an analysis model that regresses $y$ on $X$ will never be able to adjust for omitted confounder bias. Similarly, the RSR estimator is derived from an analysis model that places a hard constraint on the spatial variable in the outcome regression to be geometrically orthogonal to the exposure, thereby precluding the possibility that it can be a confounder and thus failing to adjust for it. The GP regression model from which the GLS estimator is derived is also misspecified apriori, as the GP prior for the spatial variable is modeled to be independent (statistically orthogonal) of the exposure. However, the prior specification is only a soft constraint, and when endowing the GP with a universal kernel such as the Mat\'ern, the prior gives positive mass to neighborhoods of any continuous function, including functions that are empirically correlated with the sample paths of the exposure. This enables the GLS estimator to successfully adjust for an unmeasured smooth confounder given enough data, and so a misspecified analysis model can still adjust for spatial confounding.

We confirm the consistency result via simulation experiments in Supplemental Section \ref{sec:sim}. Our theoretical results also agree with the conclusions of \cite{khan2}, who provided expressions for the finite-sample bias of the unadjusted OLS estimator (similar to the one in Lemma \ref{prop:ols}) and the GLS estimator and studied how the numerator and the denominator of the bias terms behave with respect to the smoothness of the exposure and the confounder. They generally concluded that the GLS estimator mitigates bias compared to the OLS estimator, but did not study the asymptotic properties of either one. Our asymptotic results complement the finite-sample analytical study and the extensive numerical experiments of \cite{khan2} in showing that the GLS estimator is indeed consistent under a wide variety of data generation scenarios with spatial confounding where the OLS estimator is inconsistent.

Consistency is only the first step. Future studies will focus on how the choice of the working covariance matrix, relative to the true data covariance, determines the efficiency of the estimators. The numerical experiments also reveal that in finite samples, even some of these consistent spatial estimators may lead to poor coverage as their variance estimates based on specific analysis models are not robust to model misspecification. Developing valid variance or interval estimates for effect estimators under spatial confounding is an important future direction. 

We note that there is a growing inventory of novel approaches to spatial confounding \citep{marques,schnell,papadogeorgou,dupont,thaden,keller, guan,gilbert}. Most approaches assume a linear exposure-outcome relationship, as we do here, and many rely on additional assumptions about the data generation process. For example, \cite{schnell} assumes a specific Markov structure on the joint distribution of the exposure and confounder, a ring graph design, and normality. The robustness of these estimators under relaxed assumptions needs to be studied.  
We have limited our attention to the partially linear DGP, but of course in practice the parametric restrictions encoded in the PLM outcome model may not hold;
\cite{gilbert} propose a robust nonparametric, causal inference framework for understanding and mitigating spatial confounding under minimal parametric assumptions. %
Many of these new estimators might be preferable to the traditional GLS estimator, as they are designed for mitigating spatial confounding and likely possess better second-order properties. 
We do not study these new estimators here as the goal here was to clarify contradictory claims about the first-order performance of traditional spatial  estimators like OLS, RSR, and GLS.

\section*{Acknowledgments}
 This work is partially supported by National Institute of Environmental Health Sciences (NIEHS) grant R01 ES033739 and ONR grant N00014-21-1-2820. The authors are grateful for the use of the facilities at the Joint High Performance Computing Exchange (JHPCE) in the Department of Biostatistics, Johns Hopkins Bloomberg School of Public Health that have contributed to the research results reported within this paper. We also thank Professor James S. Hodges for helpful discussions on this topic and for providing extensive feedback on an earlier version of the draft.

\clearpage 
\supplementstart

\section{Derivation of spatial confounding criterion for partially linear models}\label{sec:confder}

We describe how condition (\ref{eq:confcor}) defining spatial confounding in the PLM (\ref{eq:plm}) arises from the traditional causal inference machinery. %
Let $Y^{(x)}(S)$ denote the potential outcome at the exposure value $X(S)=x$ at random location $S$. %
Since (\ref{eq:plm}) completely specifies the outcome model, the distribution of $Y^{(x)}(S)$ for any $x$ and $S$ can simply be obtained by substituting $s=S$ and $X(S)=x$ in (\ref{eq:plm}).  

The regression coefficient $\beta^*=(\beta^*_1,\ldots,\beta^*_p)^\top$ in (\ref{eq:plm}) can be seen as the causal effect of $X(S)$ on $Y(S)$ while accounting for spatial effects captured through the nuisance function $g(\cdot)$. More formally, $\beta^*_i$ is the \textit{average treatment effect (ATE)} when increasing the $i^{th}$ component of the exposure by $1$ unit, i.e., letting $e_i$ denote the $i^{th}$ column of a $p \times p$ identity matrix, we have 
\begin{align*}\label{eq:ate}
    E[Y^{(x+e_i)}(S) - Y^{(x)}(S)] %
    =& E\big[E(Y^{(x+e_i)}(S)) - E(Y^{(x)}(S)) \given S\big] \\
    =&  E\big[(x+e_i)^\top\beta^* + g(S) - x^\top\beta^* - g(S)) \big]  \\
    =& e_i^\top\beta^*\\
    =& \beta^*_i.
\end{align*}

We say there is spatial confounding if the collection of potential outcomes $\{Y^{(x)}(S)\}_{x \in \mathbb R^p}$ are not independent of $X(S)$ marginally, but independence holds conditional on the unmeasured confounder $U=g(S)$. This is commonly referred to as \textit{conditional ignorability.} %

In our setup, as the outcome is linear in the exposure, for estimation of $\beta$, the independence in the definition of ignorability can be replaced with being uncorrelated. Thus spatial confounding occurs in the partially linear model (\ref{eq:plm}) when 
\begin{equation}\label{eq:conf}
    \cov(Y^{(x)}(S),X(S)) \neq 0, \mbox{ and } \cov(Y^{(x)}(S),X(S) \given U=u) =0 \mbox{ for all } x \in \mathbb R^p, u \in \mathbb R. 
\end{equation}
The conditional covariance between $Y^{(x)}(S) = x^\top\beta + U + \eps(S)$ and $X(S)$ at $U=u$ is simply 

\begin{align*}
    \cov(x^\top\beta + u + \eps(S),X(S) \given U=u) =& 
    \cov(\eps(S),X(S) \given U=u) \\
    =& \cov( E(\eps(s) \given U=u,S=s), E(X(s) \given U=u,S=s))  \\
    & \quad +  E( \cov( \eps(s), X(s) \given U=u,S=s)).
\end{align*}

As $\eps(s)$ is an iid process with zero mean and is independent of all other variables, we have $E(\eps(s) \given U=u,S=s)=E(\eps(s))=0$ and $\cov( \eps(s), X(s) \given U=u,S=s)=0$, and so both terms in the above equation are zero. Hence, the second condition in (\ref{eq:conf})  holds. The first condition in (\ref{eq:conf}) is equivalent to 
\begin{align*}
    &\cov(Y^{(x)}(S),X(S)) 
    \neq 0 \\
    &\iff \cov(Y^{(x)}(S),X(S)) \neq 0 \\
    &\iff \cov(x^\top\beta + U + \eps(S), X(S)) \neq 0 \\
    &\iff \cov(g(S) + \eps(S), X(S)) \neq 0 \mbox{ (as $x^\top\beta$ is a constant)} \\
    &\iff \cov(g(S), X(S)) +  \cov(\eps(S), X(S)) \neq 0 \\
    &\iff \cov(g(S), X(S)) +  E(\cov(\eps(s), X(s) | S=s)) + \cov(E(\eps(s)), E(X(s) | S=s)) \neq 0\\
    &\iff \cov(g(S),X(S)) \neq 0, \quad \mbox{ (as $\eps(s) \perp X(s)$ and $E(\eps(s))=0$ )}. %
\end{align*}
This establishes (\ref{eq:confcor}). Note that the derivation of (\ref{eq:confcor}) does not rely on whether the $g(\cdot)$ function is fixed or random, so, the definition of spatial confounding as given in (\ref{eq:confcor} holds for both scenarios.

If $g(\cdot)$ is a fixed function, the covariance condition in (\ref{eq:confcor}) can be expressed as
 \begin{align*}
    \cov(g(S),X(S)) =& E (\cov(g(S),X(S) \given S)) + \cov(E(g(S) \given S) ,E (X(S) \given S)) \\
    =& 0 + \cov(g(S) ,E (X(S))) \\
    =& \int_\calR E(X(s)) g(s) f_s(s) ds - \int_\calR E(X(s)) f_s(s) ds \int_\calR g(s) f_s(s) ds. \\
\end{align*} 
Here the second equality holds as conditional on $S$, $g(S)$ is not random, implying $\cov(g(S),X(S) \given S)=0$ and $E(g(S) \given S) = g(S)$. 

If $g(\cdot)$ is random, the covariance calculation in (\ref{eq:confcor}) will additionally account for the randomness in $g(\cdot)$. In this case, if $c_{gX}(s) = \cov(g(s),X(s))$ denotes the cross-covariance function between $g(\cdot)$ and $X(\cdot)$, then $E (\cov(g(S),X(S) \given S)) = E(c_{gx}(S)) = \int_\calR c_{gX}(s)f_s(s)ds$, and the covariance becomes:
 \begin{align*}
    \cov(g(S),X(S)) = & E(c_{gx}(S)) +  \int_\calR E(X(s)) E(g(s)) f_s(s) ds \\
    & \quad - \int_\calR E(X(s)) f_s(s) ds \int_\calR E(g(s)) f_s(s) ds. \\
\end{align*}

\section{Proof of Proposition \ref{prop:plm_id}}\label{plm_id_proof}

\noindent \textit{`Only if' part:} Let there exist some non-zero $a \in \mathbb R^p$ such that $x(\cdot)^\top a \in \calF$. %
Then, %
letting $\beta' = \beta^* - a$ and $g'(s) = g(s) + x(s)^\top a $, we have %
\[ Y(s) = x(s)^\top (\beta^* - a) + (g(s) + x(s)^\top a(s)) + \eps(s) = x(s)^\top \beta' + g'(s) + \eps(s).\]
Note that as $x(\cdot)^\top a \in \calF$ and $\calF$ is closed under addition, we also have $g'(\cdot) \in \calF$. So, for two choices of the parameters $(\beta^*,g(\cdot))$ and $(\beta',g'(\cdot))$ the PLM (\ref{eq:plm}) yields the same distribution of $Y(s)$ for all $s 
\in \calR$, and therefore the model is not identifiable for the parameters $
\beta^*$ and $g(\cdot)$.\\

\noindent \textit{`If' part:} We have $x(\cdot)^\top a \notin \mathcal{F}$ for all nonzero $a \in \mathbb R^p$. We prove the result by contradiction. If there exists two sets of parameters $(\beta^*,g(\cdot))$ and $(\beta',g'(\cdot))$ with $\beta^*,
\beta' \in \mathbb R^p$ and $g(\cdot), g'(\cdot) \in \calF$ that both yield the PLM (\ref{eq:plm})  for all $Y(s), s 
\in \calR$.
Then we have $E_{\beta^*,g(\cdot)}(Y(s)) = E_{\beta',g'(\cdot)}(Y(s))$ for all $s \in \calR$, implying
$x(s)^\top \beta^* + g(s) = x(s)^\top \beta' + g'(s)$ for all $s$, or equivalently
\[ g(s) - g'(s) = x(s)^\top a \mbox{ for } a = \beta' - \beta^* \mbox{ and all } s \in \calR. \]
The left side of the equation above is in $\calF$ as $\calF$ is closed under addition and scalar multiplication, and the right side is not in $\calF$ for any nonzero $a$ by hypothesis. Therefore, $a = \*0$, which implies that $(\beta^*,g(\cdot))=(\beta',g'(\cdot))$.

\section{Bias of OLS: Proof of Proposition \ref{prop:ols}}\label{supp:ols}

The DGP for the outcome is $Y=X\beta^* + g + \epsilon$. Without loss of generality, we can assume $g(\cdot)$ is centered in the sense that $\int_\calR g(s)f_s(s) ds=0$. This is because if it is not centered, then we can center it and add an intercept to the DGP for $Y(\cdot)$ and then include the intercept column in the design matrix $X$ for the OLS estimator. The bias result will remain the same. 

The error of the OLS estimator is given by 
\begin{align*}
\hat \beta_{OLS} - \beta^* = (X'X)^{-1}X'Y - \beta^* = (X'X)^{-1}X'g + (X'X)^{-1}X'\epsilon.
\end{align*}

Note that as $S_i \sim_{i.i.d}$ 
and $h(\cdot)$ is a fixed function, then $h(S_i)$ is i.i.d. Also $\eta(\cdot)$ is an i.i.d process. So it is immediate that $X(S_i)=h(S_i) + \eta(S_i)$ are also identically distributed. Independence can be proved in many ways. For any $A, B \subseteq \mathbb R^p$ 
\begin{align*}
    P(X(S_i) \in A, X(S_j) \in B) =& \int_{s_i,s_j} P(X(s_i) \in A, X(s_j) \in B) f_s(s_i) f_s(s_j) ds_i ds_j \\
    =& \int_{s_i,s_j,s_i \neq s_j} P(h(s_i)+\eta(s_i) \in A, h(s_j)+\eta(s_j) \in B) f_s(s_i) f_s(s_j) ds_i ds_j \\
    =& \int_{s_i,s_j,s_i \neq s_j} P(h(s_i)+\eta(s_i) \in A) P(h(s_j)+ \eta(s_j) \in B) f_s(s_i) f_s(s_j) ds_i ds_j \\
    =& \int_{s_i,s_j} P(X(s_i) \in A) P( X(s_j) \in B) f_s(s_i) f_s(s_j) ds_i ds_j \\
    =& \int_{s_i} P(X(s_i) \in A) f_s(s_i) ds_i \int_{s_j} P( X(s_j) \in B) f_s(s_j)  ds_j \\
    =&\, P(X(S_i) \in A) P(X(S_j) \in B).
\end{align*}
Here the second and fourth equality holds because the sampling distribution is continuous on $\calR$ (with a sampling density $f_s$). So the probability of selecting the same location is $0$. The third equality occurs as for fixed locations $s_i \neq s_j$, $h(s_i)$ and $h(s_j)$ are non-random and $\eta(s_i) \perp \eta(s_j)$ (as $\eta(\cdot)$ is an i.i.d. process). 

So $X(S_i)$ is marginally i.i.d.. By similar argument $Y(S_i)$ is i.i.d and pairs of variables are jointly i.i.d., e.g., $(X(S_i),g(S_i))$ is i.i.d. and $(X(S_i),\eps(S_i))$ is i.i.d. 

As $h$ is a bounded function (continuous on a compact domain) and $\var(\eta)$ is finite and strictly positive definite,  %
 $X(S_i) = h(S_i) + \eta(S_i)$ has finite second moment, i.e., 
 \begin{equation}\label{eq:vx}
     V_X := E\left(X(S)X(S)^\top\right),  \mbox{ where } S \sim f_s,
 \end{equation} is well-defined. 
So, by the strong law of large numbers $X^\top X / n = \frac 1n \sum_i X(S_i)X(S_i)^\top \to E(X(S)X(S)^\top) = V_X \succ O$, where $S\sim f_s$.

We also have $\eps(\cdot) \perp X(\cdot)$, $E(\eps(s))=0$ and $\var(\eps(s)) < \infty$, we have $\frac 1n X'\eps \to \cov(X(S),\eps(S)) = 0$, and as $g(\cdot)$ is centered we have $\frac 1n X'g \to \cov(X(S),g(S)) = v_{Xg} \neq 0$ as there is spatial confounding (\ref{eq:confcor}). 

So by the continuous mapping theorem, as $V_X$ is strictly positive definite, we have \[\hat \beta_{OLS} - \beta^* \to V_X^{-1}v_{Xg}.\] Once again, as $V_X \succ O$ and $v_{Xg} \neq 0$, we have $V_X^{-1}v_{Xg} \neq 0$, proving the nonvanishing error of the OLS estimator under spatial confounding.

\section{Consistency of GLS: Proof of Theorem \ref{thm:gls}}\label{sec:gls_consist}

We first introduce some definitions. %
Let $\calC(\calR)$ denote the space of all continuous real-valued functions on the spatial domain $\calR$, equipped with the supremum norm, i.e., for $f(\cdot) \in \calC(\calR)$, $\|f(\cdot)\|_\infty = \sup_{s \in \calR} |f(s)|$. For a given sampling density $f_s$, let $\calL_2(f_s)$ denote the Hilbert space $\left\{f(\cdot): \calR \to \mathbb R\, \Big|\,  \int_\calR f^2(s)f_s(s)ds < \infty\right\}$ equipped with the inner product $<f,g>_{\calL_2(f_s)} = \int_\calR f(s)g(s)f_s(s)ds < \infty$ and the corresponding norm $\|f\|_{\calL_2(f_s)}$. Recall that for a function $f(\cdot)$, we define $f=(f(S_1),\ldots,S_n))^\top$. We define the following norms: $\|f\|^2_2 = \sum_{i=1}^n f(s_i)^2$ and $\|f\|^2_n = \frac 1n \sum_{i=1}^n f(s_i)^2$.

\begin{proof}
We first prove the result when $g(\cdot)$ and each component of $h(\cdot)$ belongs to the RKHS $\calH_K$ of the kernel $K(\cdot,\cdot)$. Subsequently, we will extend to the case where $g(\cdot)$ and $h(\cdot)$ are just continuous functions. By the Mercer representation of the RKHS of Gaussian processes \citep[Theorem 4.51,][]{steinwart2008support}, $\calH_K$ contains functions that are scaled linear combinations of the eigenfunctions of $K(\cdot,\cdot)$. Formally,
\begin{equation}\label{eq:mercerrep}
\calH_K = \left\{ f(\cdot): f(\cdot) = \sum_{i=1}^\infty \alpha_i \lambda_i^{1/2} \phi_i(\cdot) \mbox{ with } \|f(\cdot)\|_{\calH_K}^2 := \sum_{i=1}^\infty \alpha_i^2 < \infty \right\}.
\end{equation}
Here $\{ \phi_i(\cdot) \}$ denotes the set of orthonormal eigenfunctions of $K(\cdot,\cdot)$ with respect to the sampling density $f_s$, i.e., $\int_{\calR} \phi_i(s)\phi_j(s)f_s(s)ds = \delta_{ij}$ where $\delta$ is the Kronecker delta,  $\{\lambda_i\}$ are the corresponding eigenvalues, and $\|\cdot\|_{\calH_K}$ denotes the RKHS norm. 

If $g \in \calH_K$, we can write 
\begin{equation}\label{eq:rkhsfunc}
    g(s) = \sum_{i=1}^\infty \alpha_{g,i} \lambda_i^{1/2} \phi_i(s),  \mbox{ with }  \sum_{i=1}^\infty \alpha_{g,i}^2 < \infty.
\end{equation} 

We have %
\begin{align}\label{eq:glsbias}
\begin{split}
\hat \beta_{GLS} - \beta^* &= (\bX^\top \bSigma^{-1}\bX)^{-1}\bX^\top \bSigma^{-1}\bY  - \beta^* \\
 &= (\bX^\top \bSigma^{-1}\bX)^{-1}\bX^\top \bSigma^{-1}(\bX\beta^* + \bg + \beps) - \beta^* \\
  &=(\frac 1n \bX^\top \bSigma^{-1}\bX)^{-1}(\frac 1n X^\top \bSigma^{-1}\bg + \frac 1n X^\top \bSigma^{-1}\beps). \\
\end{split}
\end{align}

The key step for the proof is to show that terms of the form $\frac 1n \bX^\top \bSigma^{-1}g$ are $o_p(1)$ for a smooth function $g$. %
Using the sub-multiplicative property of the $\|\cdot\|_2$ norm we have
\begin{equation}\label{eq:multineq}
    \|\frac 1n \bX^\top \bSigma^{-1}g\|_2 \leq \|\frac 1{\sqrt n} \bX^\top \bSigma^{-1/2}\|_2 \|\frac 1{\sqrt n} \bSigma^{-1/2} g\|_2 =  \sqrt{\|\frac 1n \bX^\top \bSigma^{-1} \bX\|_2 \left(\frac 1n g^\top \bSigma^{-1} g\right)}.
\end{equation} 
It suffices to prove that $\frac 1n g^\top \bSigma^{-1} g$ is $o_p(1)$ for a smooth $g$ and $\frac 1n \bX^\top \bSigma^{-1} \bX$ is $O_p(1)$. %

 Consider a truncated version $g_L(\cdot)$ of $g(\cdot)$ including only the first $L$ terms, i.e., 
\begin{equation}\label{eq:truncfunc}
    g_L(s) = \sum_{i=1}^L \alpha_{g,i} \lambda_i^{1/2} \phi_i(s).
\end{equation}
As $g(\cdot), g_L(\cdot) \in \calH_K$, we have $(g - g_L)(\cdot) \in \calH_K$ and $\|(g - g_L)(\cdot) \|^2_{\calH_K} = \sum_{i=L+1}^\infty \alpha_{g,i}^2$. Note that for any $s \in \calR$,  $K(s,s)=\sigma^2 > 0$ (as $K(\cdot,\cdot)$, being a Mat\'ern or squared-exponential kernel, is stationary). As $\sum_{i=1}^\infty \alpha_{g,i}^2 < \infty$, for any given $\eps > 0$, we can choose $L
\mbox{ such that } \|(g_L - g)(\cdot)\|^2_{\calH_K} \leq \eps\tau^2/\sigma^2$. By the reproducing property of $\calH_K$, for any $s \in \calR$ we have $|(g-g_L)(s)| = |< K(s,\cdot),(g-g_L)(\cdot)>_{\calH_K}| \leq \|K(s,\cdot)\|_{\calH_K}\|(g-g_L)(\cdot)\|_{\calH_K} = \sqrt{K(s,s)} \|(g-g_L)(\cdot)\|_{\calH_K} \leq \sigma * \sqrt \eps \tau/\sigma = \sqrt \eps \tau$, and  %
we have %
\begin{equation}\label{eq:rkhsdom}
    \|(g - g_L)(\cdot)\|^2_{\calL_2(f_s)} = \int_\calR (g(s)-g_L(s))^2 f_s(s)ds \leq \eps \tau^2. 
\end{equation}

We first consider the term $g_L^\top \Sigma^{-1} g_L$. Note that $\Sigma = K + \tau^2 I$, where $K_{ij} = K(S_i,S_j)$. Using Mercer's expansion of covariance functions \citep[Theorem 4.49 of][]{steinwart2008support}, we can write
\begin{equation}\label{eq:mercer}
    K(s_i,s_j) = \sum_{i=1}^\infty \lambda_i \phi_i(s_i)\phi_i(s_j) \mbox{ for all } s_i,s_j \in \calR.
\end{equation}
Consider the truncated kernel, 
\begin{equation}\label{eq:mercer}
    K_L(s_i,s_j) = \sum_{i=1}^L \lambda_i \phi_i(s_i)\phi_i(s_j) \mbox{ for all } s_i,s_j \in \calR, 
\end{equation}
and define the $n \times n$ matrix $K_L=(K_L(S_i,S_j))$ and $\Sigma_L = K_L + \tau^2 I$. Recall that we use the notation $\phi_i=(\phi_i(S_1),\ldots,\phi_i(S_n))^\top$. We have $K=\sum_{i=1}^\infty \lambda_i \phi_i \phi_i^\top$ and $K_L=\sum_{i=1}^L \lambda_i \phi_i \phi_i^\top$.  
As $\lambda_i \geq 0$ for all $i$, for any vector $u \in \mathbb R^n$, we have $u^\top K u = \sum_{i=1}^\infty \lambda_i (u^\top \phi_i)^2 \geq \sum_{i=1}^L \lambda_i (u^\top \phi_i)^2 = u^\top K_L u \geq 0$, i.e., $K \succcurlyeq K_L \succcurlyeq O$. Consequently, as $\tau^2 > 0$, we have $\Sigma \succcurlyeq \Sigma_L \succ O$. So, 
\begin{equation}\label{eq:truncineq}
    0 \leq g_L^\top \Sigma^{-1} g_L \leq g_L^\top \Sigma_L^{-1} g_L. 
\end{equation} 
 Let $\Phi=(\phi_1,\ldots,\phi_L)$, $D_L = \diag(\lambda_1,\ldots,\lambda_L)$ and $a_{g,L}=(\alpha_{g,1},\ldots,\alpha_{g,L})^\top$. Then 
\begin{equation}\label{eq:vector}
    g_L = \Phi_L D_L^{1/2} a_{g,L} \mbox{ and } \Sigma_L = \Phi_L D_L \Phi_L^\top + \tau^2 I. 
\end{equation}
Using the Sherman-Woodbury-Morrison identity, we have 
\begin{equation}\label{eq:lowrankinv}
     \Sigma_L^{-1} = \tau^{-2} I - \tau^{-4} \Phi_L( D_L^{-1} + \tau^{-2} \Phi_L^\top \Phi_L)^{-1} \Phi_L^\top. 
\end{equation}
Let $M_L = \frac 1n \Phi_L^\top \Phi_L$. Then we have
\begin{align*}%
    \frac 1n g_L^\top\Sigma_L^{-1}g_L =& \frac 1n  a_{g,L}^\top D_L^{1/2} \Phi_L^\top \left(\tau^{-2}I - \tau^{-4} \Phi_L( D_L^{-1} + \tau^{-2} \Phi_L^\top \Phi_L)^{-1} \Phi_L^\top \right) \Phi_L D_L^{1/2} a_{g,L}\\
    =& \frac 1n \tau^{-2} a_{g,L}^\top D_L^{1/2} \Phi_L^\top \Phi_L D_L^{1/2} a_{g,L} \\
    & \qquad - \frac 1n \tau^{-4} a_{g,L}^\top D_L^{1/2} \Phi_L^\top \Phi_L( D_L^{-1} + \tau^{-2} \Phi_L^\top \Phi_L)^{-1} \Phi_L^\top \Phi_L D_L^{1/2} a_{g,L}\\
    =& \tau^{-2} a_{g,L}^\top D_L^{1/2} M_L D_L^{1/2} a_{g,L} \\
    & \qquad -  \tau^{-4} a_{g,L}^\top D_L^{1/2} M_L( D_L^{-1}/n + \tau^{-2} M_L)^{-1} M_L D_L^{1/2} a_{g,L}\\
     =& \tau^{-2} a_{g,L}^\top D_L^{1/2} M_L (I - \tau^{-2}  (D_L^{-1}/n + \tau^{-2} M_L)^{-1} M_L) D_L^{1/2} a_{g,L}.
\end{align*}
All vectors and matrices in the above expression are at most $L$-dimensional. For a fixed $L$, as $n \to \infty$, and $S_i$ are i.i.d., we have by the law of large numbers. 
\begin{align*}
    M_L(j,k)=\frac 1n \sum_{i=1}^n \phi_j(S_i)\phi_k(S_i) \to \int_\calR \phi_j(s)\phi_k(s) f_s(s)ds = \delta_{jk}. 
\end{align*}
Hence, $M_L \to I$ as $n \to \infty$ where $I$ here is the $L \times L$ identity matrix. The other terms $a_{g,L}$ and $D_L$ %
do not depend on $n$. So $D_L^{-1}/n \to O$ and we can calculate
\begin{align*}
    \lim_{n\to \infty}\, \frac 1n g_L^\top\Sigma_L^{-1}g_L =& \tau^{-2} a_{g,L}^\top D_L^{1/2} I (I - \tau^{-2}  (O + \tau^{-2} I)^{-1} I) D_L^{1/2} a_{g,L} \\
    =& \tau^{-2} a_{g,L}^\top D_L^{1/2} I (I - \tau^{-2} \tau^2 I) D_L^{1/2} a_{g,L} \\
    =&0.
\end{align*}
Consequently, applying (\ref{eq:truncineq}) we have 
\begin{align}\label{eq:limglsqf}
    \lim_{n\to \infty} \frac 1n g_L^\top\Sigma^{-1}g_L = 0.
\end{align}
Applying Cauchy-Schwartz and the AM-GM inequality, we have
\begin{align*}
     \frac 1n g^\top\Sigma^{-1}g \leq&  \frac 1n g_L^\top\Sigma^{-1}g_L +  \frac 1n (g-g_L)^\top\Sigma^{-1}(g-g_L) + 2 |\frac 1n g_L^\top\Sigma^{-1}(g-g_L)|\\
     \leq& \frac 1n g_L^\top\Sigma^{-1}g_L +  \frac 1n (g-g_L)^\top\Sigma^{-1}(g-g_L) + 2 \sqrt{\frac 1n g_L^\top\Sigma^{-1}g_L  \frac 1n (g - g_L)^\top\Sigma^{-1}(g-g_L)}\\
    \leq& 2\frac 1n g_L^\top\Sigma^{-1}g_L + 2 \frac 1n (g-g_L)^\top\Sigma^{-1}(g-g_L).
\end{align*}
As $n \to \infty$, the first term on the right-hand side is $0$ from (\ref{eq:limglsqf}). As eigenvalues of $\Sigma^{-1}$ are bounded from above by $1/\tau^2$, the second term is less than $1/\tau^2 \|(g-g_L)\|^2_n$. As $S_i$ are i.i.d., by the law of large numbers and (\ref{eq:rkhsdom}), 
\begin{align*}
    \|g-g_L\|^2_n =  \frac 1n \sum_{i=1}^n (g-g_L)(S_i)^2 \to \int_\calR (g-g_L)^2(s) f_s(s)ds = \|(g-g_L)(\cdot)\|^2_{\calL_2(f_s)} \leq \eps  \tau^2,  
\end{align*}
and we obtain $\lim_{n\to \infty} \frac 1n g^\top\Sigma^{-1}g \leq 1/\tau^2 \eps \tau^2 = \eps$. Since $\eps > 0$ is arbitrary, we have for any $g(\cdot) \in \calH_K$
\begin{equation}\label{eq:limglsqfwhole}
    \lim_{n \to \infty} \frac 1n g^\top \Sigma^{-1} g = 0. 
\end{equation}

Returning to (\ref{eq:multineq}), let $S$ denote a random variable such that $S \sim f_s$. As $X(S_i)$ are marginally i.i.d. (see proof of Proposition \ref{prop:ols}), we have $\frac 1n X^\top \Sigma^{-1} X \preccurlyeq 1/\tau^2 \frac 1n X^\top X \to 1/\tau^2 E(X(S)X(S)^\top) = 1/\tau^2 V_X$, i.e., $\frac 1n X^\top \Sigma^{-1} X = O_p(1)$. 
As $\frac 1n g^\top \Sigma^{-1} g$ is $o_p(1)$ by \ref{eq:limglsqfwhole}) and $\frac 1n X^\top \Sigma^{-1} X = O_p(1)$ we have from (\ref{eq:multineq}) that 
$\frac 1n X^\top \Sigma^{-1} g = o_p(1)$.  

In the expression (\ref{eq:glsbias}) of the bias of the GLS estimator, the cross term $\frac 1n X^\top \bSigma^{-1}\beps$ is also $o_p(1)$, as $\eps(\cdot)$ is an i.i.d error process independent of $X(\cdot)$. Lemma \ref{lem:cross} provides a general result and formal proof of this. So the error of the GLS estimator is $o_p(1)$ as long as we can prove that $(\frac 1n X^\top \Sigma^{-1} X)^{-1}$ is $O_p(1)$, which we do below. 

We can write 
\begin{equation}\label{eq:qf}
    A_n := \frac 1n X^\top \Sigma^{-1}X =  \frac 1n \bh^\top \bSigma^{-1}\bh + \frac 1n \boeta^\top \bSigma^{-1}\boeta + \frac 2n \bh^\top \bSigma ^ {-1}\boeta.
\end{equation}

Due to the independence of the error terms $\eta$ with all other variables, the cross-term $\frac 2n \bh^\top \Sigma ^ {-1}\boeta$ is once again $o_p(1)$ by Lemma \ref{lem:cross}. As $h(\cdot)$ is bounded (being a continuous function on a compact domain), and eigenvalues of $\Sigma^{-1}$ are bounded by $1/\taus$, it follows that $\frac 1n \bh^\top \bSigma^{-1}\bh = O_p(1)$. Also, similarly as $\eta(S_i)$ are i.i.d. with finite second-moment, 
\begin{equation}\label{eq:upperboundeta}
    \frac 1n \boeta^\top \bSigma^{-1}\boeta \prec \frac 1\taus \sum_{i=1}^n \eta(S_i)\eta(S_i)^\top \to \frac 1\taus E(\eta(S)\eta(S)^\top) \implies \frac 1n \boeta^\top \bSigma^{-1}\boeta = O_p(1). 
\end{equation}
As $A_n$ is fixed ($p \times p$)-dimensional and the cross-term in (\ref{eq:qf}) is $o_p(1)$, we have $A_n -B_n \to 0$ where $B_n =  (\frac 1n \bh^\top \Sigma ^ {-1}\bh + \frac 1n \boeta^\top \Sigma ^ {-1}\boeta)$. By (\ref{eq:upperboundeta}) and the argument preceding it, both $A_n$ and $B_n$ are $O_p(1)$. Since $A_n \to B_n$ and both are fixed-dimensional $O_p(1)$ matrices, we have $\det(A_n) \to \det(B_n)$ (as the determinant is simply a polynomial of the entries of the matrix). Since $\frac 1n \bh^\top \Sigma ^ {-1}\bh \succcurlyeq {O}$, we have $\det(B_n) \geq \det(\frac 1n \boeta^\top \Sigma ^ {-1}\boeta)$ (by Lemma \ref{lem:det}) and thus $\det(A_n) > \det(\frac 1n \boeta^\top \Sigma ^ {-1}\boeta) -\upsilon$ for any $\upsilon > 0$ with probability going to $1$. Lemma \ref{lem:qf} shows that determinants of quadratic forms like $\frac 1n \boeta^\top \bSigma ^ {-1}\boeta$ are bounded away from zero (uniformly in $n$, i.e., the bound does not depend on $n$), implying that $\det(B_N)$ and consequently $\det(A_n)$ are bounded away from $0$ with probability going to $1$ (as $\upsilon$ can be chosen to be arbitrarily small), or equivalently, $1/\det(A_N)=O_p(1)$.  

Now $A_n^{-1} = \adj(A_n)/\det(A_n)$ where $\adj(A_n)$ is the $p \times p$ adjugate matrix. As $A_n$ %
is $O_p(1)$, %
and $\adj(A_n)$ is of the same fixed dimension with entries taken from $A_n$ (reordered and possibly with a change of sign), we have $\adj(A_n) = O_p(1)$.  %
So, $A_n^{-1}=(\bX^\top \bSigma^{-1} X)^{-1}$ is $O_p(1)$ and the result is proved when $g(\cdot)$ and each $h_j(\cdot)$ are in $\calH_K$.

When $g(\cdot)$ and $h_j(\cdot)$ are merely continuous functions (not necessarily belonging $\calH_K$), the only thing that needs to be proved is that $\frac 1n g^\top \Sigma^{-1} g =o_p(1)$, and similarly for the $h_j(\cdot)$. When $K(\cdot,\cdot)$ is the Mat\'ern or squared exponential kernel, the RKHS $\calH_K$ approximates any continuous function in $\calC(\calR)$ arbitrarily closely in the supremum norm \citep[Example 1 of][]{vaart}. Such kernels are called `\textit{universal kernels}' \citep{micchelli2006universal}. So for any $\delta > 0$, there is a $g_{\calH_K}(\cdot) \in \calH_K$ such that $\max_{s \in \calR} |g(s) - g_{\calH_K}(s)| < \delta.$ Thus, once again using Cauchy-Schwartz and AM-GM inequality, we have

\begin{align*}
    \frac 1n g^\top \Sigma^{-1} g \leq 2\left(\frac 1n g_{\calH_k}^\top \Sigma^{-1} g_{\calH_k} + \frac 1n (g-g_{\calH_k})^\top \Sigma^{-1} (g-g_{\calH_k}) \right).
\end{align*}

The first term on the right side is $o_p(1)$ by (\ref{eq:limglsqfwhole}) as $g_{\calH_K}(\cdot) \in \calH_K$. As eigenvalues of $\Sigma^{-1}$ are bounded from above by $1/\taus$, the second term is less than $1/\taus \frac 1n \sum_i (g(S_i) - g_{\calH_K}(S_i))^2 \leq \delta^2 / \taus.$
As $\delta > 0$ is arbitrary, the result follows. 
The same argument holds for each $h_j(\cdot)$. \\

\end{proof}

\begin{lemma}\label{lem:cross} Let $S_i \sim_{i.i.d}, i=1,2,\ldots$ denote locations sampled randomly from a domain $\calR \in \mathbb R^d$ using a sampling density $f_s$. Let $w(S_i)$, $i=1,2,\ldots$ denote i.i.d $p$-dimensional random vectors with finite second moment and let $w=(w(S_1),\ldots,w(S_n))^\top$. Let $\varepsilon(s)$ be i.i.d. zero-mean finite-variance  random variables with for $s \in \calR$, with $\varepsilon(\cdot)$ independent of $w$ and $S_1,S_2,\ldots$ Define $\varepsilon=(\varepsilon(S_1),\ldots,\varepsilon(S_n))^\top$. Let $\bSigma$ denote an $n \times n$ covariance matrix, depending only on $S_1,\ldots, S_n$ and some fixed parameters, with eigenvalues uniformly bounded (in $n$) from below. %
Then $\frac 1n w^\top  \bSigma^{-1}\varepsilon$ is $o_p(1)$.\\
\end{lemma}

\begin{proof}
It is sufficient prove the result for $p=1$ and for $p > 1$ apply that result to each entry of the $p$-dimensional vector $\frac 1n w^\top  \bSigma^{-1}\varepsilon$. Let $\calS=\{S_1,\ldots,S_n\}$ and $S \sim f_s$ denote a random location. %
Also let $E(w^2(S))=\chi$ and $\var(\varepsilon(s)) = \iota^2$ for any $s \in \calR$. These quantities are finite by the statement of the lemma. The eigenvalues of $\bSigma^{-1}$ are bounded from above by some  $T>0$. Then $\frac 1n w^\top  \bSigma^{-1} \varepsilon = \frac 1n w^\top  Q D Q^\top  \varepsilon$ where $Q D Q^\top$ is the spectral decomposition of $\Sigma^{-1}$, $Q$ has orthonormal columns and $D=\diag(d_i)$ is diagonal. Let $u = (u_1,\ldots,u_n)^\top= Q^\top  w$ and $v =(v_1,\ldots,v_n)^\top=  Q^\top  \varepsilon$. Since $Q$ is orthonormal, $\varepsilon(s)$ are i.i.d., independent of $\calS$ and $w$, and with zero mean, and $Q,D$ are fixed matrices conditional on $\calS$, we have $$E(v \given w, \calS) = E(Q^\top E (\varepsilon \given w, \calS)) =  E(Q^\top E (\varepsilon)) = 0, \mbox{ and }$$
$$ \var(v \given w,\calS) = Q^\top \var(\varepsilon \given w,\calS) Q =  Q^\top \var(\varepsilon) Q = \iota^2 I.$$
 Finally, as $Q$ is orthonormal, %
 $\|u\| = \|w\|$. %

Now, the quantity of interest is $w^\top \Sigma^{-1} \varepsilon =  u^\top  D v = \sum_i d_i u_i v_i$. %
Clearly, $E (\sum_i d_i u_i v_i) = E[E (\sum_i d_i u_i v_i \mid w,\calS)] = E [\sum_i d_i u_i E(v_i \mid w, S) ] = 0 $. We can write its variance as follows:

\begin{align*}
\var (\sum_i d_i u_i v_i) =& \var[E(\sum_i d_i u_i v_i \mid  \calS, w] + E[\var(\sum_i d_i u_i v_i \mid  \calS, w)]\\
=& 0 + E[ (\sum_i d_i^2 u_i^2) \iota^2] \quad \mbox{ (as } E(v \given w,\calS)=0, \var(v \given w,\calS) = \iota^2 I\mbox{)}\\ 
\leq& \iota^2 T^2 E(\sum_i u_i^2)\\
=& \iota^2 T^2 E(\sum_i w(S_i)^2) \quad \mbox{ (as } \|w\| = \|u\|\mbox{)}\\
=& n \iota^2 T^2 \chi.
\end{align*}

Putting this together,

\begin{align*}
\lim_n \var(\frac 1n w^\top  \bSigma^{-1} \varepsilon ) =& \lim_n  \var(\frac 1n u^\top  \bD v)\\
=& \lim_n \frac{1}{n^2} \var(u^\top \bD v)\\
\leq& \lim_n \frac 1n T^2 \iota^2 
\chi \\
=&\; 0
\end{align*}

Since $\frac 1n w^\top  \Sigma^{-1} \varepsilon $ has mean zero and limiting variance zero, it converges to zero in probability by Chebyshev's inequality.\\
\end{proof}

\begin{lemma}\label{lem:qf} Let $S_i \sim_{i.i.d}, i=1,\ldots,n$ denote locations sampled randomly from a domain $\calR \in \mathbb R^d$ using a sampling density $f_s$. Let $\bSigma$ denote an $n \times n$ covariance matrix with constant diagonal entries and eigenvalues uniformly bounded (in $n$) from below. Also let $\boeta_i \sim_{i.i.d}$ random variables in $\mathbb R^p$ such that $\var(\eta_i)$ is finite and strictly positive definite, and denote  $\boeta=(\boeta_1,\ldots,\boeta_n)^\top$. Then $1/\det(\frac 1n\boeta^\top \Sigma^{-1}\boeta)$ is $O_p(1)$.\\ %
\end{lemma}

\begin{proof}
Let $\eta_i=(\eta_i^{(1)},\ldots,\eta_i^{(p)})^\top$ for $i=1,\ldots,n$ and $\eta^{(j)}=(\eta^{(j)}_1,\ldots,\eta^{(j)}_n)^\top$ for $j=1,\ldots,p$. Let the scalar $D$ equal the constant diagonal entries of $\Sigma$ ($D$ exists by the statement of the Lemma), i.e., $D=\Sigma(i, i)$. As the eigenvalues of $\Sigma$ are uniformly bounded away from $0$, let $T$ denote a lower bound of the GP eigenspace of $\Sigma$, i.e., $0 < T \leq \lambda_i(\Sigma)$ for all $i$ and $n$, $\lambda_i$ denoting the $i^{th}$ eigenvalue of $\Sigma$. %
If $B=\frac 1 {\sqrt n} (\Sigma^{1/2}\boeta ,\, \Sigma^{-1/2}\boeta)$ then $B^\top B$ is positive semidefinite and $\boeta^\top  \Sigma \boeta/n \geq T \boeta^\top \boeta / n$ is positive definite with probability going to $1$ (as $\var(\eta_i) \succ O$). 

The rest of this argument holds with probability going to $1$, conditioning on the set where $\boeta^\top \boeta / n$ (and, thus $\boeta^\top \Sigma \boeta / n$) are positive definite. 

As $B^\top B$ and its top-left block  $\boeta^\top \Sigma \boeta / n$ are positive definite, the Schur complement $\Psi = \frac 1n\boeta^\top \Sigma^{-1}\boeta - \frac 1n \boeta^\top \boeta(\frac 1n \boeta^\top \Sigma\boeta)^{-1}\frac 1n \boeta^\top \boeta$ is positive semidefinite. So by Lemma \ref{lem:det}, 

\[\det(\frac 1n \boeta^\top \Sigma^{-1}\boeta) \geq \det(\Psi) + \det(\frac 1n \boeta^\top \boeta(\frac 1n \boeta^\top \Sigma\boeta)^{-1}\frac 1n \boeta^\top \boeta) \geq \det(\frac 1n \boeta^\top \boeta(\frac 1n \boeta^\top \Sigma\boeta)^{-1}\frac 1n \boeta^\top \boeta).\]

Equivalently, we have
\[\frac 1{\det(\frac 1n \boeta^\top \Sigma ^ {-1}\boeta)} \leq \frac{1}{\det(\frac{1}{n}\boeta^\top \boeta(\frac 1n \boeta^\top \Sigma\boeta)^{-1}\frac 1n \boeta^\top \boeta)}\]\[ =  \frac{\det(\boeta^\top \Sigma\boeta/n)}{\det(\boeta^\top \boeta/n)^2} \]
The inverse of the denominator is immediately $O_p(1)$ as the $\eta_i$ are i.i.d and $\var(\eta_i) \succ O$. For the numerator, let $M_n = \boeta^\top \Sigma \boeta/n$. We immediately have $E [M_n(j,j)] = E(\eta^{(j)\top} \Sigma \eta^{(j)}/n) = \rho_{j} \tr(\Sigma)/n = \rho_{j} D$  where $\rho_{j}$ is the variance of $\eta^{(j)}$. Therefore, by Markov's inequality, each $M_n(j,j)$ is $O_p(1)$ and by Hadamard's inequality, $\det(M_n) \leq \prod_{j=1}^p M_n(j,j)$ and is $O_p(1)$. 
\\
\end{proof}

The following fact, which is essentially a weak corollary of the Brunn-Minkowski determinant inequality \citep[see, e.g.,][ for a modern treatment]{yuan2007generalization} that will be useful for our results.

\begin{lemma}\label{lem:det}
Let $A$ and $B$ denote positive semidefinite matrices. Then $\det(A+B) \geq \det(A) + \det(B)$. 
\end{lemma}

\begin{proof} First we consider the case where $A, B$ are positive definite. 
Then the matrix form of the Brunn-Minkowski inequality \citep{yuan2007generalization} gives
\begin{equation}
    \det(A+B)^{1/p} \geq \det(A)^{1/p} + \det(B)^{1/p}
\end{equation}
where $p$ is the dimension of $A$ and $B$. As $\det(A) > 0$ and $\det(B) > 0$, it follows that
\begin{equation}
    \det(A+B) \geq (\det(A)^{1/p} + \det(B)^{1/p})^p \geq \det(A) + \det(B).
\end{equation}

This proves the case where $A, B$ are strictly positive definite. If $A,B$ are positive semidefinite but not necessarily positive definite, define $A(\eps) = A + \eps I$ and $B(\eps) = B+\eps I$. For $\eps>0$, $A(\eps), B(\eps)$ are strictly positive definite, so we have 
\begin{equation}\label{eq:det}
    \det(A(\eps)+B(\eps)) \geq \det(A(\eps)) + \det(B(\eps)). 
\end{equation} 
Since the determinant is simply a polynomial function of the matrix entries, we can write
\[   \lim_{\eps \to 0^+} \det(A(\eps)+B(\eps)) \geq  \lim_{\eps \to 0^+} [\det(A(\eps)) + \det(B(\eps))]
\]
\[
    \det(A+B) \geq \det(A) + \det(B).
\]

\end{proof}

\section{Conditional and marginal consistency, and Proof of Theorem \ref{thm:endo}}\label{cond_marg}

\subsection{Conditional consistency implying marginal consistency}\label{cond_consist_conv_proof}

We first present a general result on when estimators that are consistent conditional on realizations of some random variables, are also consistent marginally (i.e., over multiple draws of that random variable). This result is then used to prove the consistency of GLS when the spatial terms in the outcome and exposure are random functions.\\

\begin{proposition}\label{prop:cond_consist}
Let $\hat\beta$ be an estimator whose distribution depends on random variables $U$ and $V$ with $U \perp V$. If there exists a set $\mathcal{U}$ such that $\pr(U \in \mathcal{U}) = 1$ and for all $u \in \mathcal{U}$, $\hat\beta$ converges in probability to $\beta^*$ conditional on $U=u$, then $\hat\beta$ converges unconditionally to $\beta^*$.\\
\end{proposition}

\begin{proof} %
We will write $\hat \beta = \hat \beta(U,V)$ whenever the dependence of $\hat \beta$ on $U$ and $V$ needs to be made explicit. Let $F_U$ denote the distribution of $U$ and $\calU$$(\delta, n) = \{u: \forall N \geq n : \pr(| \hat\beta(u,V) - \beta^* | > \delta ) < \delta  \}$.  Then 

\begin{equation}\label{eq:cond_to_marg}
\begin{split}
 \limsup_n \pr [|\hat \beta - \beta^*| > \delta ] &= \limsup_n \{ \pr[(U \in \calU(\delta, n)) \cap (|\hat \beta - \beta^*| > \delta)] + \\
 & \qquad \pr[U \notin \calU(\delta, n) \cap (|\hat \beta - \beta^*| > \delta) ]   \}\\
 & \leq \limsup_n \int_{U \in \calU(\delta,n)} \pr [|\hat \beta - \beta^*| > \delta \mid  U=u] F_U(du) + \\
 & \qquad \limsup_n \pr[U \notin \calU(\delta, n)]. %
 \end{split}
\end{equation}

The first term on the right hand side of (\ref{eq:cond_to_marg}) can be expressed as 
\begin{align*}
    & \limsup_n \int_{U \in \calU(\delta,n)} \pr [|\hat \beta(u,V) - \beta^*| > \delta \mid  U=u] F_U(du) \\
    & \qquad = \limsup_n \int_{U \in \calU(\delta,n)} \pr [|\hat \beta(u,V) - \beta^*| > \delta] F_U(du) \\
    & \qquad \leq \delta.
\end{align*} 

Here the equality occurs as $U \perp V$, and the inequality occurs from the definition of $\calU(\delta,n)$. 

Next we consider the second term $\limsup_n \pr[U \notin \calU(\delta, n)]$ in the right side of (\ref{eq:cond_to_marg}).

 Note that the sets $\calU(\delta,n)$ are increasing in $n$. So 
\[\limsup_n \pr[U \in \calU(\delta, n)] = \lim_n \pr[U \in \calU(\delta, n)] =  \pr[\lim_n \calU(\delta, n)].\]
Also, recall that the set $\calU$ of $u$ for which $\hat \beta(u,V)$ converges in probability to $\beta^*$ has probability one.
If $u \in \calU$, then $\hat \beta(u,V) \to \beta^*$ which implies that for any $\delta$, there is some $n_0$ for which $ \forall N \geq n_0 : \pr[| \hat\beta(u,V) - \beta^* | > \delta] < \delta $. So, $u \in \calU$ implies $ u \in \calU(\delta,n)$ for large enough $n$, i.e., $u \in \lim_n \calU(\delta,n)$. As $P(\calU)=1$, we then have $P(\lim_n \calU(\delta,n)) = 1$.

Returning to (\ref{eq:cond_to_marg}), therefore, for any $\delta > 0 $

\[
\limsup_n [\pr (|\hat \beta - \beta^*| > \delta) ] \leq \delta + 0 = \delta.
\]

Since $\delta >0$ was arbitrary,  $\hat \beta \to \beta^*$ in probability.
\end{proof}

\subsection{Proof of Theorem \ref{thm:endo}}

We apply Proposition \ref{prop:cond_consist} with $U=(g(\cdot),h(\cdot))$, $V=(S_1,S_2,\ldots,\eta(\cdot),\eps(\cdot))$, and $\calU=\{\omega \given (g(\cdot)(\omega),h(\cdot)(\omega)) \mbox{ is continuous }\}$. As sample paths of $(g(\cdot),h(\cdot))$ are almost surely continuous, we have $P(\calU)=1$. Also, by Theorem \ref{thm:gls}, $\hat\beta_{GLS}(u,V)$ is consistent for any draw $u=(g(\cdot)(\omega),h(\cdot)(\omega))$ which are continuous fixed functions. Hence, the conditions of Proposition \ref{prop:cond_consist} are met, and the result follows.

\subsection{Marginal consistency implies conditional consistency}\label{marg_consist_conv_proof}

As an aside, we note that the converse of Proposition \ref{prop:cond_consist} is also true, as given in Proposition \ref{prop:cond_consist_conv}.

\begin{proposition}\label{prop:cond_consist_conv}
Let $\hat\beta(U,V)$ be an estimator whose distribution depends on random variables $U$ and $V$ with $U \perp V$. Assume that $\hat\beta$ converges unconditionally to $\beta^*$. Then there exists a set $\mathcal{U}$ such that for all $u \in \mathcal{U}$, $\hat\beta(u,V)$ converges in probability to $\beta^*$ , and $\pr(U \in \mathcal{U}) = 1$.\\ %
\end{proposition}

\begin{proof} We prove by contradiction. Let $F_U$ denote the distribution of $U$. Assume there is a set $\mathcal{H}$ such that $p=\pr(U \in \mathcal{H})>0$, and for all $u \in\mathcal{H}$, $\hat\beta(u,V)$ does \textit{not} converge to $\beta^*$. Then there exists a functional $\epsilon(u)>0$ such that for all $u\in\mathcal{H}$: $\limsup_n  \pr(|\hat\beta(u,V)-\beta^*| > \epsilon(u) ) > 0$. Let $\epsilon^*$ be a median of $\epsilon(u)$, over the conditional distribution of $U$ within $\mathcal{H}$. Clearly, $\epsilon^*>0$ since $\epsilon(u)$ is strictly positive. Let $H^* = \{u: (u \in \mathcal{H}) \cap (\epsilon(u) \geq \epsilon^*)\}$. By construction, $\pr(U \in H^*) \geq \frac{p}{2}$. So we have

\begin{align*}
\limsup_n \pr[|\hat\beta-\beta^*| > \epsilon^* ] &\geq  \limsup_n \pr[(|\hat\beta(U,V)-\beta| > \epsilon^*) \cap  (U \in H^* )]\\
& = \limsup_n \int_{H^*} \pr[|\hat\beta(u,V)-\beta| > \epsilon^* \given U=u] d(F_U)\\
& = \limsup_n \int_{H^*} \pr[|\hat\beta(u,V)-\beta| > \epsilon^*] d(F_U). \quad (\mbox{ as } U \perp V)\\
\end{align*}

As $\limsup_n \pr[|\hat\beta(u,V)-\beta| > \epsilon^*]$ is strictly positive for $u \in H^*$, and $\pr[U \in H^*] \geq p/2$, the quantity above is strictly positive, 
contradicting the assumption that $\hat\beta$ converges to $\beta^*$ in probability. Thus no set $\mathcal{H}$ can exist, and the set of $u$ for which $\hat\beta(u,V)$ does not converge to $\beta$ has measure zero.

\end{proof}

\section{Simulations}\label{sec:sim}

\subsection{Fixed/Random confounder}\label{sec:sim_bias}
We demonstrate that a spatial model is needed to account for omitted variable bias and that the distinction between fixed and random confounding functions is not crucial for the performance of the spatial estimators.

 We use $n=10,000$ locations chosen uniformly at random from the square with coordinates between 0 and 10. On these points, we simulate a confounding variable $g(\cdot)$ from a Gaussian process with spherical covariance $\widetilde \Sigma(\cdot,\cdot)$ with marginal variance $\sigma^2=1$ and inverse range $\phi = 0.25$ \citep[using the \texttt{BRISC} R-package,][]{saha}. Note that this GP is used to simulate a smooth fixed function $g(\cdot)$ and is unrelated to the GP covariance function used in the GLS estimator. We run $1,000$ replications. In the first scenario, we generated only one \textit{fixed confounder} surface $g(\cdot)$, which is maintained throughout all runs while regenerating $X$ and $Y$ in each replicate run. 
 The confounder is smoothed by a thin plate spline with $200$ basis functions. 
 We summarize the data generation process below. 
\begin{equation}\label{eq:dgpfixed}
    \begin{split}
            S_i \sim_{i.i.d}\, \mbox{Unif}(0,10) \times \mbox{Unif}(0,10), D_{ij} = |S_{i}-S_{j}|, \\
            \widetilde \Sigma_{ij} \approx \mathbb{I}(D_{ij} \leq \phi^{-1})*\sigma^2(1- \frac{3}{2}\phi D_{ij} + \frac{1}{2} (\phi D_{ij})^2),  Z^* \sim N(0, \widetilde \Sigma_{ij}), \\
            g(\cdot) = \operatorname{spline.smooth}(Z^*),   g= (g(S_1),\ldots,g(S_n))^\top,\\
             X \sim_{ind} N(g, 1), 
            Y \sim_{ind} N(X+g, 1).
            \end{split}
    \end{equation}
So, in this scenario we have a univariate exposure $X(\cdot)$ and its spatial part $h(\cdot)$ is exactly equal to $g(c\dot)$, inducing strong confounding. Note the given correlation function $\widetilde \Sigma$ is approximate because the data are generated with the \texttt{BRISC} package \citep{saha}, which uses a nearest-neighbor approximation to generate large-scale realizations of Gaussian process.

 We estimated the coefficient $\beta$ on $X$ (whose true value $\beta^*$ equals 1, as we see in (\ref{eq:dgpfixed})) using OLS, GLS, and three different spline models.   %
Due to the large sample size, the GLS estimator is computed according to the exponential covariance function using a nearest-neighbor Gaussian process approximation \citep{nngp,finley2019efficient} implemented in the \texttt{BRISC} R-package \citep{brisc}, and the confidence interval uses the \texttt{BRISC} parametric spatial bootstrap method \citep{saha}. The OLS confidence interval is analytic. In addition, we fit three different spline models using the \texttt{mgcv} package \citep{wood}, all of which use $200$ basis functions. The first, ``GAM" is the default model with a smoothness penalty chosen by generalized cross-validation. The next, ``GAM.fx" (for ``fixed") has no smoothness penalty. Finally, ``Spatial.plus" is the Spatial+ technique of \cite{dupont}, a two-stage procedure which uses a spatially residualized version of the exposure. The GAM models return an analytic standard error from the linear model including the basis functions, which we use to construct the confidence intervals.

\begin{figure}[!h]
    \centering
    \includegraphics[trim={0 0 0 24},clip,scale=.7]{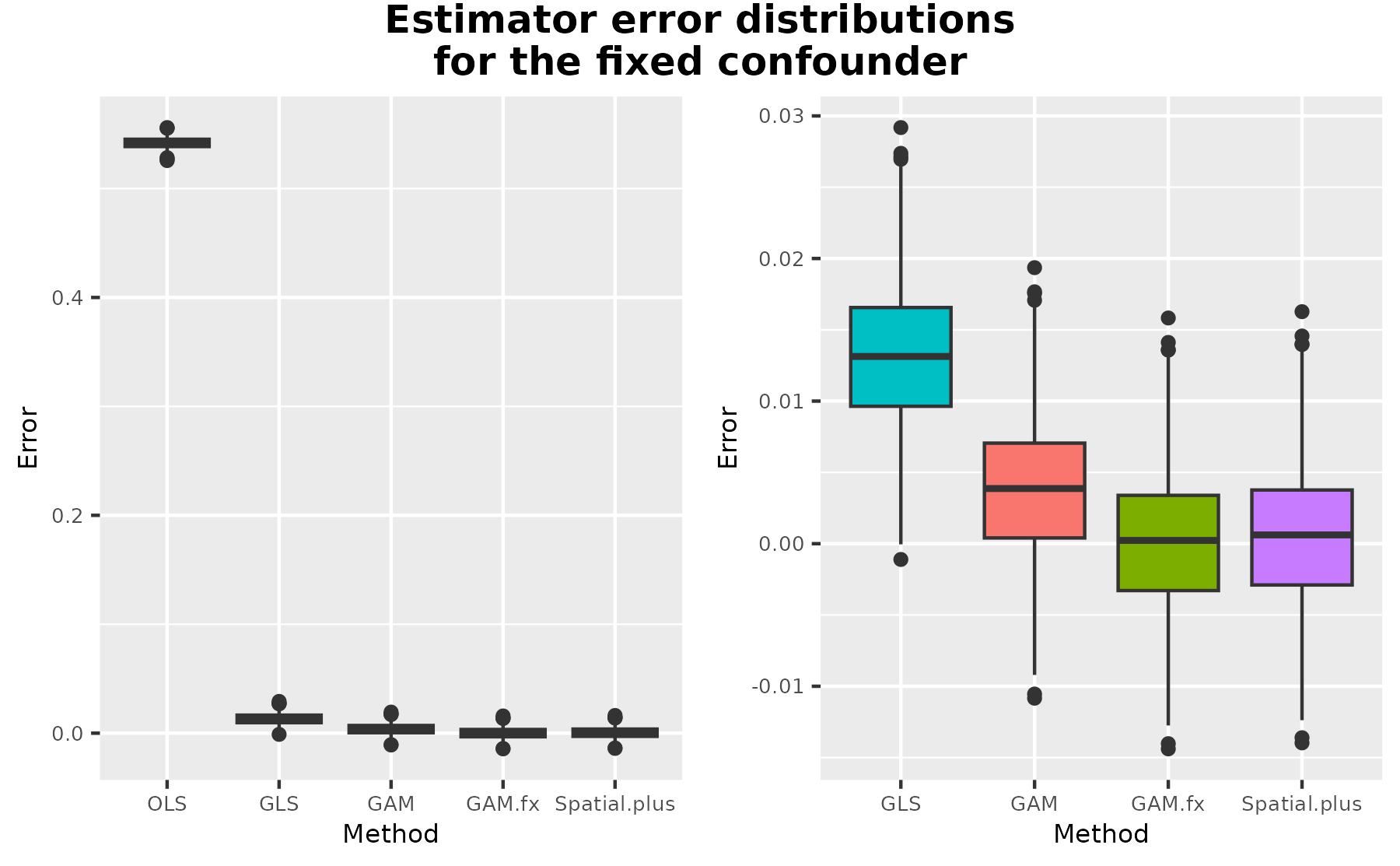}
    \caption{Sampling distributions of the bias under spatial confounding by a fixed function. For improved visibility, the right plot is a version of the left plot that removes the OLS (due to its very large bias) and zooms in on the performance of the non-OLS methods.}
    \label{fig:fixed}
\end{figure}

We find in Figure \ref{fig:fixed} (left) that, despite the strong unmeasured spatial confounding present in the DGP, all spatial methods adjust for nearly all the large bias present in the OLS estimator. Zooming into the performance of only the spatial estimators in Figure \ref{fig:fixed} (right), we see that the GAM and GLS methods still exhibit some bias which also shows up in poor confidence interval coverage in Table \ref{tab:fixed}.
\begin{table}[h]
\centering
\begin{tabular}{lrrrrr}
  \hline
  & OLS & GLS & GAM & GAM.fx & Spatial.plus \\
  \hline
 Coverage & 0\% & 25\% & 89\% & 96\% & 96\% \\
  \hline
\end{tabular}
   \caption{95\% confidence interval coverage probabilities for all estimators for the fixed confounder}\label{tab:fixed}
\end{table}
For GAM, this small bias is likely a regularization bias incurred due to the penalized fitting method used, as we see that GAM.fx, which does not penalize, does not suffer from this. It has been shown that efficient estimation of the exposure effect with the spatial term modeled as splines requires undersmoothing (less penalization) of the splines \citep{rice}. It is likely that what we see here for GAM as opposed to GAM.fx is a manifestation of this. 

For GLS we also suspect that some regularization bias is contributing to its bias and poor coverage. The GLS estimator is derived from the marginalization of the GP analysis model (\ref{eq:gp}), and \cite{yang} has shown that undersmoothing is also necessary for this GP-based PLM to obtain minimax optimal estimates of $\beta$. They recommend letting the spatial decay in the GP covariance go to infinity with sample size at a certain rate, which we do not do in our implementation as the spatial covariance parameters are estimated using maximum likelihood. Also, the true confounding function here (generated as a spline-smoothed fixed draw of a Gaussian process) is quite ``wiggly," and oversmoothing in an analysis model will lead to bias. When a smoother confounding function was used in \cite{gilbert} in a similar partially linear DGP, the GLS had near nominal coverage \citep[see Table S2 of ][]{gilbert}. Our simulations in Section \ref{sec:simcluster} using a confounder that is simply a linear function of space also yield near nominal coverage for the spatial GLS estimator (see Table \ref{table:gls_coverage}). These further points to undercoverage in this particular setting caused by oversmoothing when the spatial parameters in the working covariance matrix of GLS are estimated by MLE and where the true confounder is less smooth.

Also note that the GAM estimators' overall superior coverage (compared to GLS) may be explained by the fact that the data are truly generated according to spline basis functions, meaning that all the GAM-based models are correctly specified. On the other hand, the \texttt{BRISC} parametric bootstrap method used to generate the confidence intervals is under the assumption of dependent data with a Gaussian process covariance, which is misspecified under confounding. This demonstrates that %
model misspecification may not affect first-order properties, e.g., the consistency of the GLS point estimator, but can adversely impact its variance estimation.

The Spatial+ and GAM.fx methods perform well, and equally so, both apparently because they do not suffer from regularization bias.
This demonstrates that for methods suffering from this kind of bias, it can be beneficial to augment the outcome regression with exposure models as well, as discussed in \cite{gilbert} and \cite{dupont}. Finally, we note that the bias of GLS,  while larger than the GAM methods, is still quite small $(\sim 3\%)$ relative to the true effect size of $1$, and is an order of magnitude smaller than the bias of the OLS/RSR estimator.

\begin{figure}[h]
    \centering
    \includegraphics[trim={0 0 0 24},clip,scale=.7]{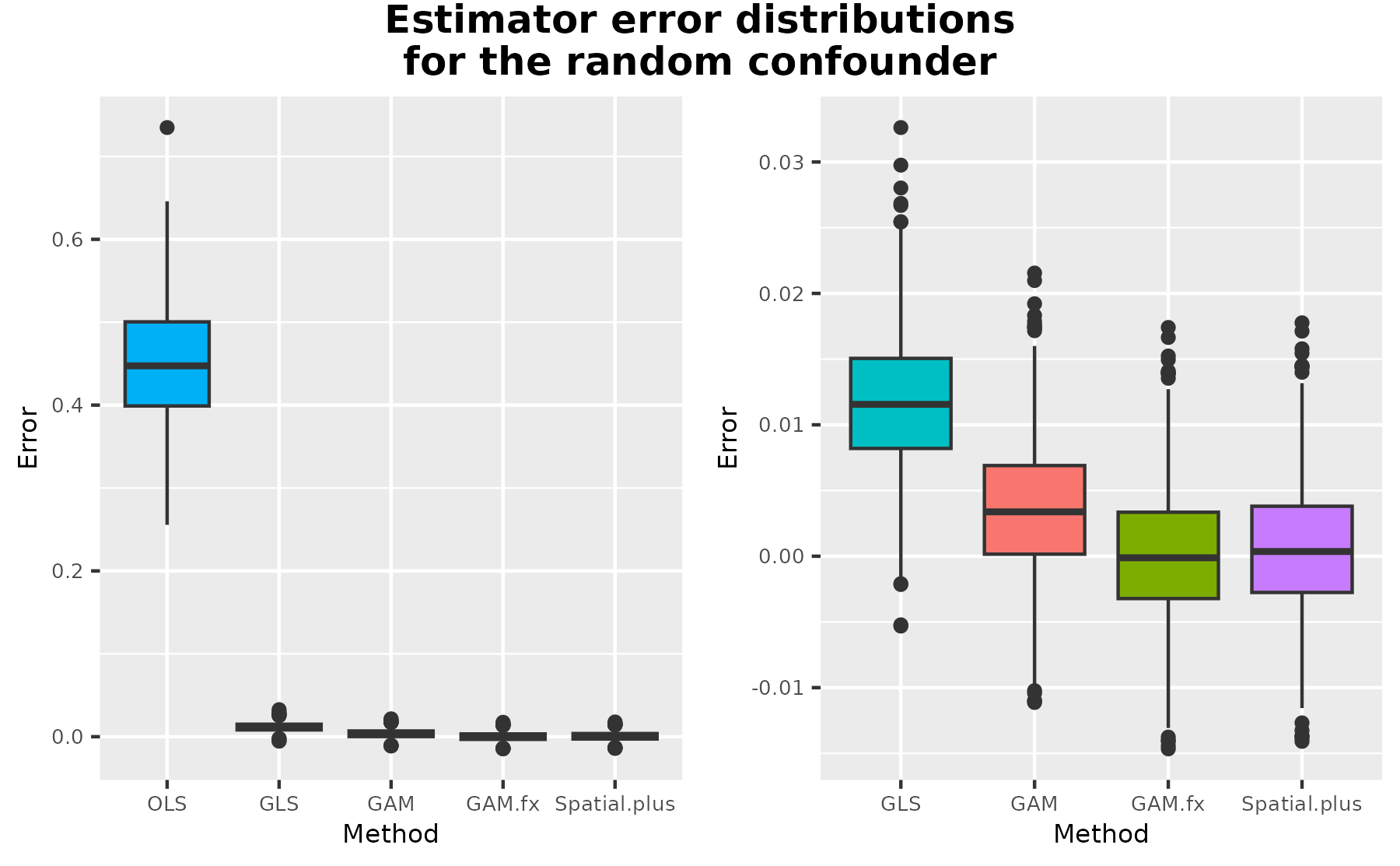}
    \caption{Sampling distributions of the bias under random spatial confounding (endogeneity). For improved visibility, the right plot zooms in on the performance of the non-OLS methods.}
    \label{fig:random}
\end{figure}

We then consider a scenario where confounding is by a random function $g(\cdot)$. We use a DGP identical to the scenario of confounding by fixed function as specified in (\ref{eq:dgpfixed}) except that the variable $Z^*$ is regenerated according to its Gaussian process distribution for each run of the simulation, resulting in a random $g(\cdot)$ (which is the spline smoothed version of $Z^*$) for every simulation. The results in Figure \ref{fig:random} and Table \ref{tab:random} are extremely similar to the fixed confounder scenario, except that the OLS estimator exhibits much larger variance, as shown by the greater width of the distribution.

\begin{table}[h]
\centering
\begin{tabular}{lrrrrr}
  \hline
  & OLS & GLS & GAM & GAM.fx & Spatial.plus \\
  \hline
 Coverage & 0\% & 35\% & 89\% & 94\% & 94\% \\
  \hline
\end{tabular}
\caption{95\% confidence interval coverage probabilities for all methods for the random confounder}\label{tab:random}
\end{table}

\newpage
\subsection{Analytical expression of finite-sample bias of GLS}\label{sec:eigen}

In this simulation study, we investigate the finite-sample performance of the GLS estimator. We use the Mercer representation of the reproducing kernel Hilbert space (RKHS) of the Gaussian process covariance kernel used in the GLS estimator to provide an approximate expression for the bias term of the GLS estimator which we then verify empirically through numerical experiments.

We consider the DGP (\ref{eq:plm}) with confounding by a fixed function $g(\cdot)$ and a one-dimensional exposure $X(\cdot)$. Conditional on the locations and the values of $X$, the exact bias of the GLS estimator $\hat\beta_{GLS} = ({X}^\top {\Sigma}^{-1} {X})^{-1} {X}^\top {\Sigma}^{-1} {Y}$ is
\begin{equation}\label{eq:exact}
\mbox{Bias}_{\mbox{exact}} = \frac{\bX^\top  \bSigma^{-1} \bg}{\bX^\top  \bSigma^{-1} \bX}.    
\end{equation}

Here $\Sigma$ is the working covariance matrix used in the GLS estimator (\ref{eq:gls})., i.e.,  %
${\Sigma} = {K} + \tau^2 {I}$, 
with the GP kernel matrix ${K}$ and nugget variance $\taus$.

We consider the case where the functions $g(\cdot)$ and $h(\cdot)$ (respectively corresponding to the outcome model (\ref{eq:plm}) and exposure model (\ref{eq:exp})) are both finite linear combinations of eigenfunctions of $K(\cdot,\cdot)$, i.e., for some integer $L$ we have

\begin{equation}\label{eq:gh}
    g(s) = \sum_{i=1}^L \alpha_{g,i}\lambda_i^{1/2} \phi_i(s), h(s) = \sum_{i=1}^L \alpha_{h,i}\lambda_i^{1/2} \phi_i(s). 
\end{equation}

As discussed in Section \ref{sec:fixed}, the RKHS $\calH_K$ of $K(\cdot,\cdot)$ approximates any continuous function, and as seen from the definition of $\calH_K$ in (\ref{eq:rkhsfunc}), it is the closure of functions of the form (\ref{eq:gh}). Hence, for large enough $L$, the functions $g$ and $h$ specified via (\ref{eq:gh}) can approximate any continuous function. 

Using Mercer's theorem, we can write $K(s,u) = \sum_{l=1}^\infty \lambda_l \phi_l(s)\phi_l(u)$ where $\lambda_l \geq 0$ are the eigenvalues (in decreasing order) and $\phi_l(\cdot)$ are the corresponding orthonormal basis functions. %
As 
$ \int K(s,u)\phi_l(u)f_s(u) du = \lambda_l \phi_l(s)\, \forall l, s.$
If $S_i \sim_{i.i.d.} f_s$, then the LHS can be approximated as a Riemann sum to obtain
$$ \frac 1n \sum_{i=1}^n K(s,S_i)\phi_l(S_i) \approx \lambda_l \phi_l(s)\,\forall l, s. \implies  %
\frac 1n \bK \bphi_l \approx \lambda_l \bphi_l, l=1,\ldots,n.$$ Letting $\bPhi=(\bphi_1,\ldots,\bphi_n)^\top $ and $\bD = \diag(\{\lambda_l \mid l=1,\ldots,n\})$, as $\frac 1n \Phi^\top \Phi \to I$, we have 
$$\bK \approx  (\Phi/\sqrt n) (n\bD) (\Phi^\top/\sqrt n)$$ as an approximate spectral decomposition of $K$, and consequently, 

\begin{equation}\label{eq:eigen}
\bSigma \approx \left(\frac 1 {\sqrt n} \bPhi\right) (n\bD + \taus \bI) \left(\frac 1 {\sqrt n} \bPhi^\top \right);\, \bSigma^{-1} \approx \left(\frac 1 {\sqrt n} \bPhi\right) (n\bD + \taus \bI)^{-1} \left(\frac 1 {\sqrt n} \bPhi^\top \right).
\end{equation}

For $g(\cdot)$ and $h(\cdot)$ as in (\ref{eq:gh}) and $\Sigma^{-1}$ approximated as in (\ref{eq:eigen}), we have the following approximate expression of the GLS bias predicted by the theoretical eigenfunction analysis: 
\begin{equation}\label{eq:biasrkhs}
\text{Bias}_{\text{pred}} = \frac{\sum_{i=1}^{L} \alpha_{g,i} \alpha_{h,i} \lambda_i / (n \cdot \lambda_{i} + \tau^2)}{\bX^\top  \bSigma^{-1} \bX/n}. %
\end{equation}

From the form of the theoretically predicted approximate bias, it is clear that the bias goes to zero as $n \to \infty$ if $L$, the number of eigenfunctions included in $\bg(\cdot)$ and $\bh(\cdot)$, and the corresponding coefficients, remain constant.
We now compare the exact and the eigen-expansion-based expressions of the bias given respectively in (\ref{eq:exact}) and (\ref{eq:biasrkhs}) of the GLS estimator in a numerical study. 

The expression in (\ref{eq:biasrkhs}) and the proof of Theorem \ref{thm:gls} relied on the eigenfunctions of the assumed covariance structure, which in practice often do not admit closed-form expressions. However, in the case of one-dimensional locations generated according to a Gaussian sampling density $f_s$, a squared-exponential covariance function \textit{does} have analytic expressions for the eigenfunctions and eigenvalues.

Our simulation is structured as follows. For $n=3,000$, we generate $n$ locations $S_1,\ldots,S_n$ from a $N(0,1)$ distribution. On these points, we calculate the squared-exponential kernel as 
\begin{equation}
K(s, s') = \exp \left( - \phi (s - s')^2 \right),
\end{equation}
where we take $\phi=1/2$. As given in \cite{rasmussen}, this kernel $K$ has eigenvalues
\begin{equation}
\lambda_k = \sqrt{\frac{2a}{A}}B^k
\end{equation}
and eigenfunctions
\begin{equation}
\phi_k(x) = \exp \left(- (c - a) x^2\right) H_k \left(\sqrt{2c} x\right)
\end{equation}
where $H_k(x) = (-1)^k \exp(x^2) \frac{d^k}{dx^k} \exp(-x^2)$ is the $k$th order Hermite polynomial, and the constants are $a^{-1} = 4\sigma^2$, $b = \phi$, $c = \sqrt{a^{2} + 2ab}$, $A = a + b + c$, $B = \frac{b}{A}$.

For $L=5$, we calculate spatial functions $\bh$ and $\bg$ as in (\ref{eq:gh}) as linear combinations of the first $L$ eigenfunctions with random coefficients as follows. Reparametrizing $c_{g,i} = \alpha_{g,i}\lambda^{1/2}$ and $c_{h,i} = \alpha_{h,i}\lambda^{1/2}$, these true coefficients are sampled independently from a standard normal distribution:
\[
c_{g, i}, c_{h,i} \sim_{i.i.d} N(0,1); \quad i=1,...,L.
\] 

Then the data are generated as 
\begin{align*}
{X} &= {h} + \eta; \quad \eta \sim \mathcal{N}(0, \kappa^2),\\
{Y} &=  {X} + {g} + \epsilon; \quad \epsilon \sim \mathcal{N}(0, \tau_0^2).
\end{align*}

We use $\kappa^2 = \frac{1}{16}$ and $\tau_0^2=1$. The small value of $\kappa^2$, making ${X}$ nearly a spatial function itself, ensures that the bias need not be close to zero in samples of limited size.

For GLS estimation, we use a working covariance matrix %
${\Sigma} = K + \tau^2 {I}$ where $K_{ij}=K(S_i,S_j)$.
Note that we take $\tau^2 \neq \tau_0^2$ to demonstrate that the nugget variance need not be known \textit{a priori}. 

We compare the exact and theoretically predicted biases. The results of the simulation are shown in Figure \ref{fig:eigensim} for $500$ trials. 
It is clear that the exact bias and the bias predicted from the theory of Section \ref{sec:fixed} are generally in good agreement. This offers an empirical verification of the theoretical argument on the consistency of the GLS estimator using Mercer representation of the RKHS of the Gaussian process kernel. 

\begin{figure}[h]
    \centering
    \includegraphics[scale=.45]{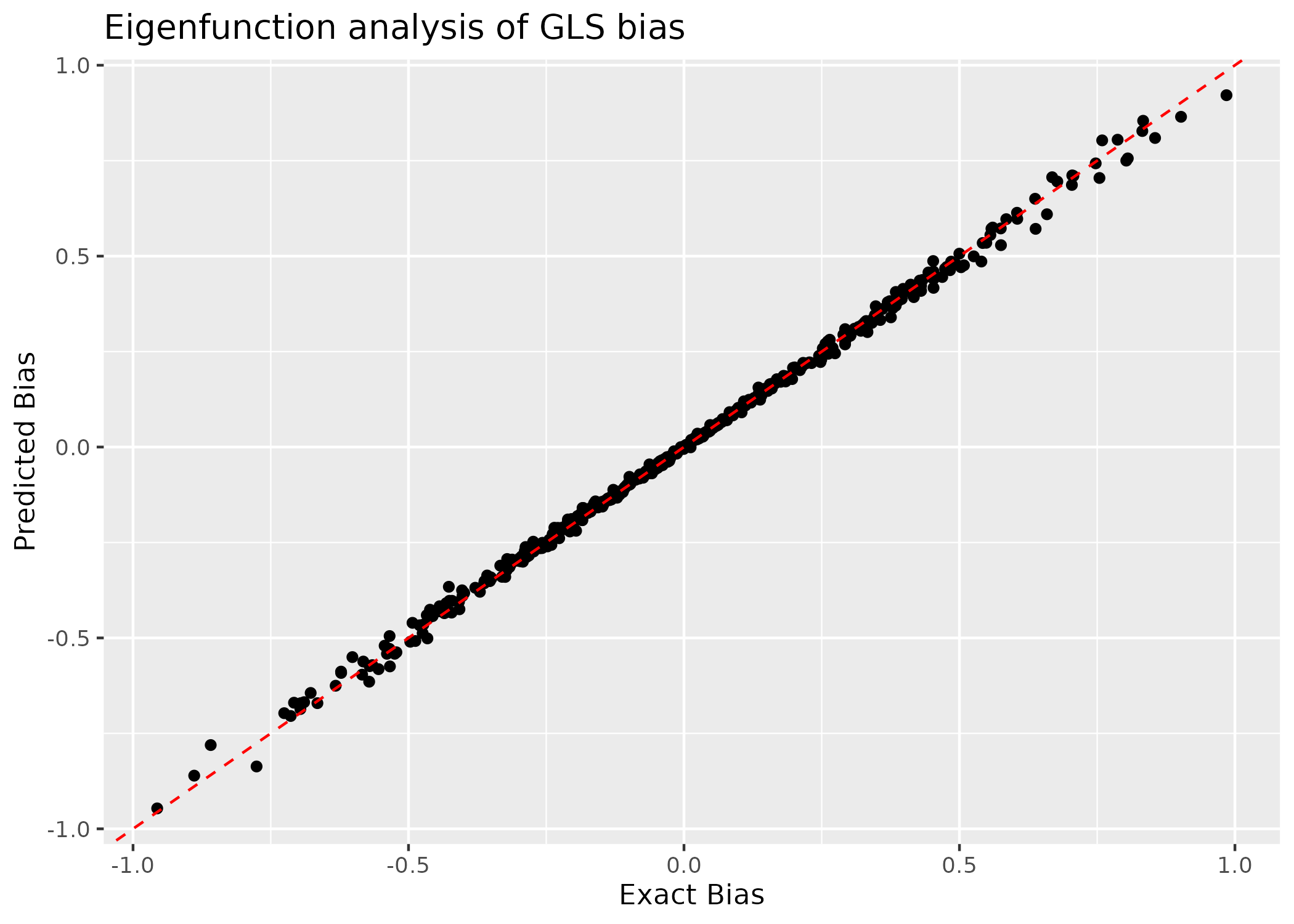}
    \caption{Comparison of the exact bias and the bias predicted from the theory of Section \ref{sec:fixed} for the GLS estimator using eigenfunction analysis of the covariance function.}
    \label{fig:eigensim}
\end{figure}

\newpage
\subsection{Role of working covariance matrix and sample size}\label{sec:simconsist}

Our theoretical results on the consistency of GLS do not require the working covariance matrix used in the GLS estimator to match the true covariance of the outcomes. We conduct a simulation study to empirically verify the robustness of the consistency result to misspecification of the working covariance matrix, and assess the change in biases of the estimators with sample size. We generate data as 
\begin{align*}
    S_i &\sim Unif([0,1]^2), i=1,\ldots,n,\\
    X(S_i) &\sim N(g(S_i),v_0^2), \mbox{ where } v^2_0=1,\\
    Y(S_i) &\sim N(X(S_i)\beta^* + g(S_i), \tau^2_0), \mbox{ where } \beta^*=1, \mbox{ and } \tau^2_0=1. 
\end{align*}

We consider three data generation scenarios: 

\begin{enumerate}[(i)]
    \item The true confounder $g(\cdot)$ is a fixed function of space, showed in Figure \ref{fig:surface} (left). The working covariance matrix for the GLS estimator is $\Sigma=K + \tau^2 I$ where $K(s,s')=\sigma^2 \exp(-\phi \|s-s'\|)$ is an exponential covariance matrix with variance $\sigma^2 = 2$ and decay (inverse range) $\phi=10$, and $\tau^2=0.25$ is the working nugget. Note that as the true $g(\cdot)$ is fixed here, the spatial structure in $y(\cdot)$ comes from $g(\cdot)$ in the mean and not from the covariance. So the working covariance matrix used in the GLS is misspecified. Even the nugget used in the working covariance matrix $\tau^2=0.25$ is different from the true nugget $\tau^2_0=1$ used to generate the $y(S_i)$'s. \\
    
    \item The true confounder $g(\cdot)$ is a random function (a GP with a squared exponential covariance function), i.e., $g \sim GP(0,K_0)$ where $K_0(s,s') = \sigma^2_0 \exp(- \phi_0 \|s-s'\|^2)$. One example (one draw) from this confounder process is given in Figure \ref{fig:surface} (right). The working covariance matrix is correctly specified, i.e., $\Sigma = K_0 + \tau_0^2 I$, with $\sigma^2=2, \phi_0=20,\tau^2_0=1$. \\
    
    \item  The true $g(\cdot)$ is the same as scenario (ii), but the working covariance is misspecified, taken to be the same choice as in scenario (i). So not only the covariance family is misspecified (true = squared exponential, working = exponential), but %
    the inverse range, and the nugget are also misspecified as %
    $\tau^2 \neq \tau^2_0$, and $\phi \neq \phi_0$. 
\end{enumerate} 

\begin{figure}[h]
    \centering
    \includegraphics[width=0.4\linewidth]{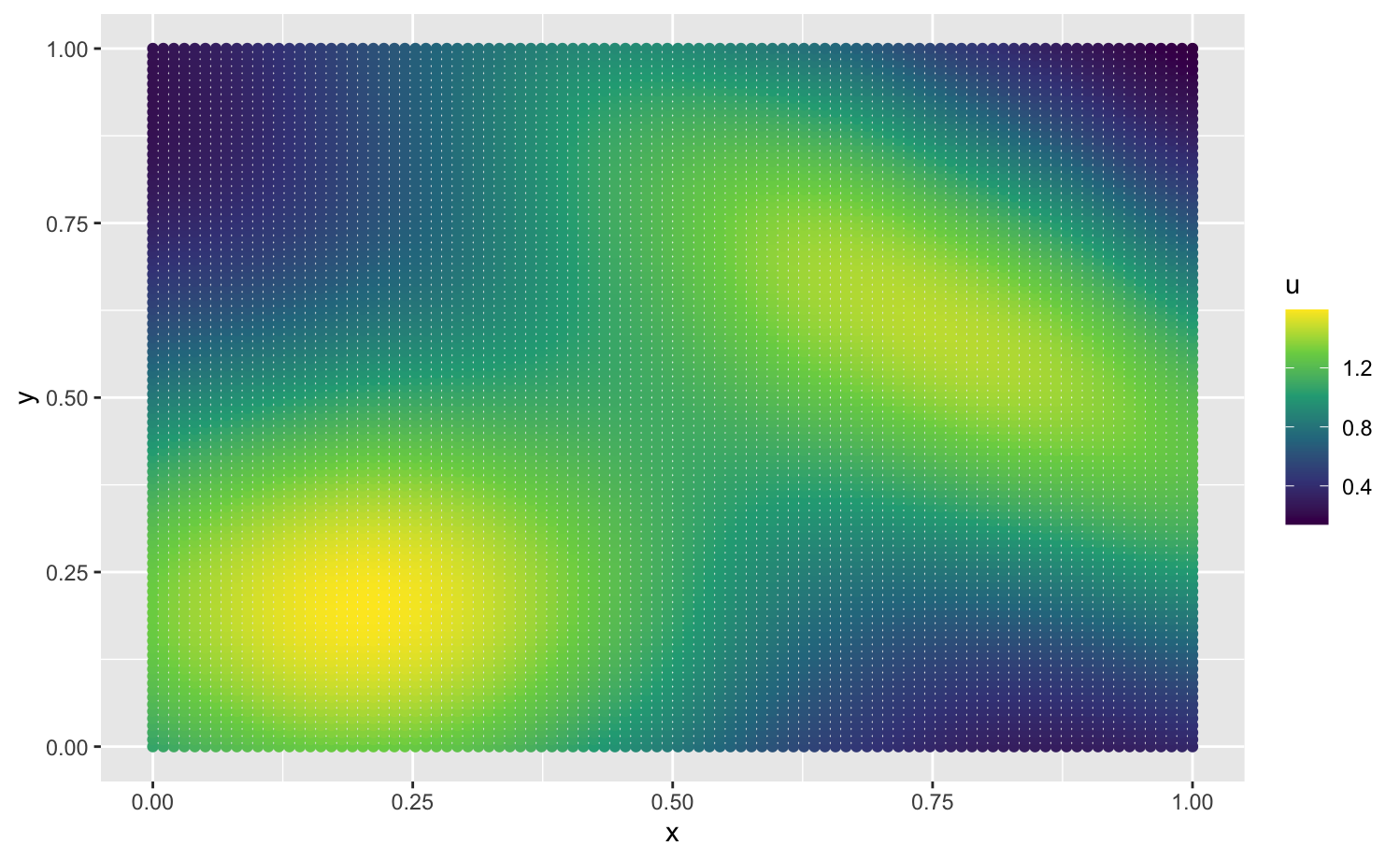}
    \includegraphics[width=0.4\linewidth]{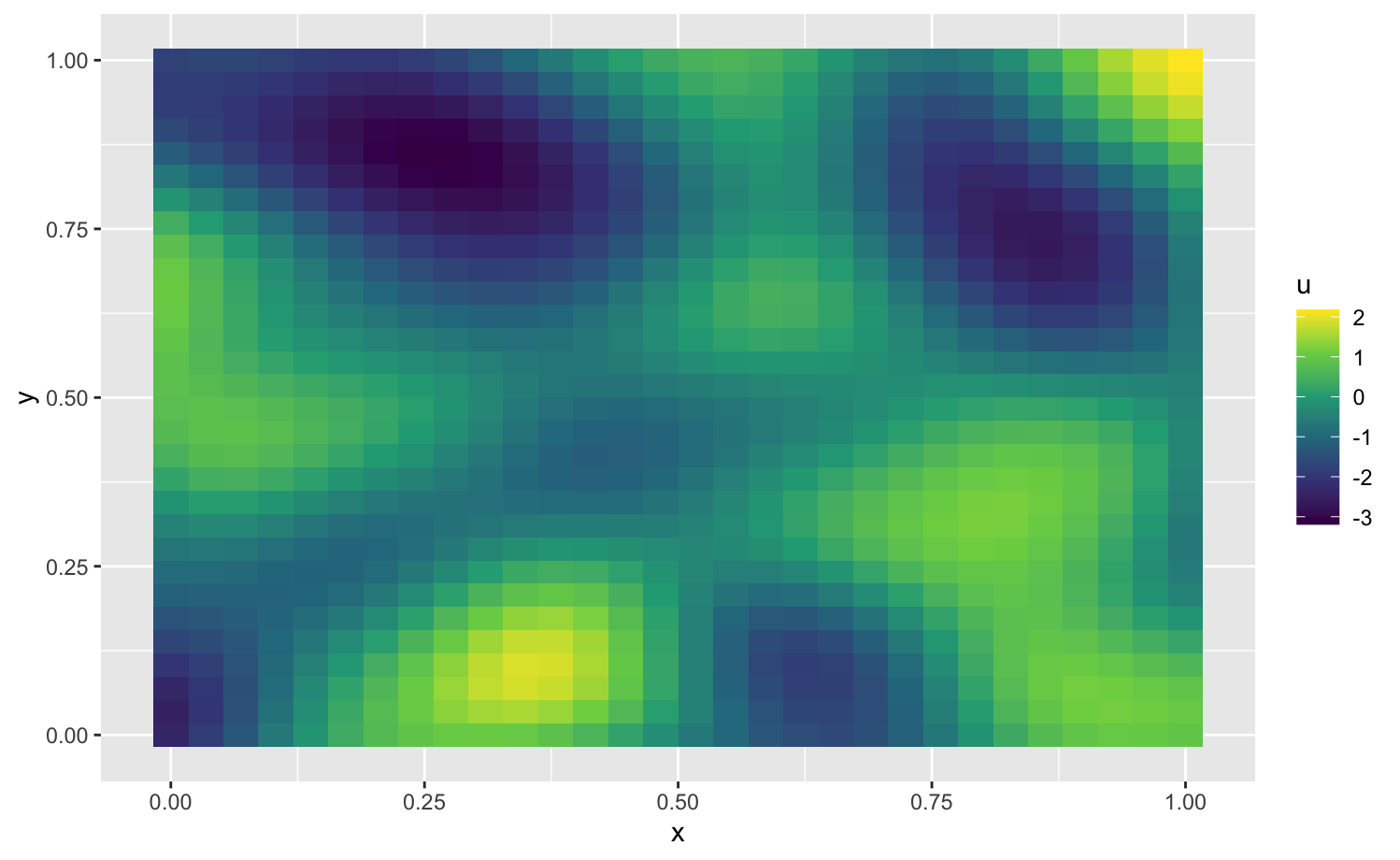}
    \caption{(Left) The fixed $g(\cdot)$ for Scenario (i), (Right) One example (for one seed) of a random $g(\cdot)$ for scenarios (ii) and (iii).}
    \label{fig:surface}
\end{figure}

    For each scenario, we ran $100$ replicate simulations and present the average absolute bias of the GLS estimator (in \%) in Figure \ref{fig:bias}. We see that for all 3 scenarios, the bias of the GLS sharply decreases to zero with sample size, while the bias of the OLS stays roughly at the same non-zero level. Comparing scenarios (ii) and (iii), we do see the benefit of using a correctly specified working covariance matrix as the biases are generally lower than the analogous numbers for the misspecified setting (particularly at smaller sample sizes). %

    \begin{figure}
        \centering
        \includegraphics[width=0.9\linewidth]{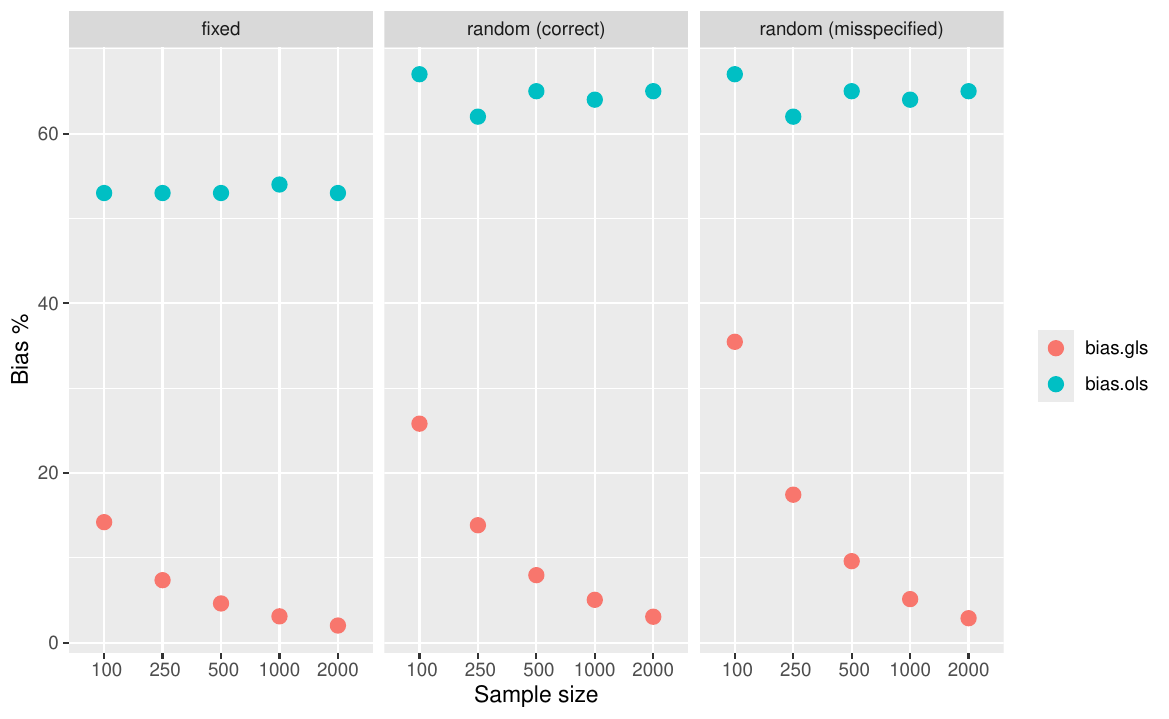}
        \caption{Bias of the GLS and OLS under spatial confounding as a function of sample size. The three panels correspond to fixed confounder (Scenario (i)), random confounder with correctly specified GLS working covariance matrix (Scenario (ii)), and random confounder with misspecified GLS working covariance matrix (Scenario (iii)).}\label{fig:bias}
    \end{figure}

\newpage
\subsection{Spatially smooth (continuous) unmeasured confounding versus  clustered (discrete) unmeasured  confounding}\label{sec:simcluster}

The following simulation assesses to what extent adding exogenous (i.e., not endogenous) spatial errors in the analysis model can adjust for spatial confounding or endogeneity in the DGP, contrary to the well-established literature showing that it cannot mitigate confounding in general. To compare performance of exogenous analysis models that exploit the spatially continuous nature of the confounder versus ones that simply treat the confounder as a discrete group-specific variable, we slightly deviate from the DGP considered in (\ref{eq:plm}), and consider a generalization of the DGP (\ref{eq:plm}) where multiple observations are recorded at each (one-dimensional) location. The sample size is $3,000$; this is achieved with $n=300$  locations and $k=10$ data points for each location. We use a confounder that is a simple linear function of spatial location in order to focus on the differences in location modeling. The DGP can be summarized as 
\begin{align*}
S_{ij} &= i, & g(s)\,&=s/10,\\
X_{ij} &\sim N(g(S_i), 1), & Y_{ij} &\sim N(X_{ij} + g(S_i), 1) \mbox{ for } i=1, \dots, 
n; j=1,\dots,k.\\
\end{align*}

 We investigate the impact of modeling location in two distinct ways by fitting two versions of GLS; one is a discrete analysis model with group-specific random effects that do not use any spatial information, i.e., 
 \begin{equation}\label{eq:glsdisc}
 Y_{ij} \sim N(X_{ij} + \nu_i, 1), \nu_i \sim_{i.i.d} N(0,\sigma^2).
 \end{equation}
 The other is an analysis model that treats location as a continuous measurement and uses a Gaussian process to model the random effects, i.e., 
  \begin{equation}\label{eq:glscont}
      Y_{ij} \sim N(X_{ij} + \nu(S_i), 1), \nu(\cdot) \sim GP(0,K(\cdot,\cdot)).
  \end{equation}
  The GLS estimator for either analysis model arises after marginalization of the random effects over their respective priors.
  
 The GLS estimator for the grouped location is computed with the location-specific independent random effects model using the \texttt{lme4} package \citep{bates}; this estimator we call ``grouped." The GLS estimator for the continuous location is computed using a Gaussian process with an exponential covariance function. However, due to the large sample size, we use a nearest-neighbor Gaussian process approximation \citep{nngp} implemented in the \texttt{BRISC} R-package \citep{saha}. This estimator we call ``spatial." Standard errors for the grouped estimator come from subsampling 1/20 of the dataset 120 times. The spatial estimator uses the \texttt{BRISC} bootstrap. The empirical distributions are based on 1,000 replications.

\begin{figure}[ht]
    \centering
    \includegraphics[scale=.4]{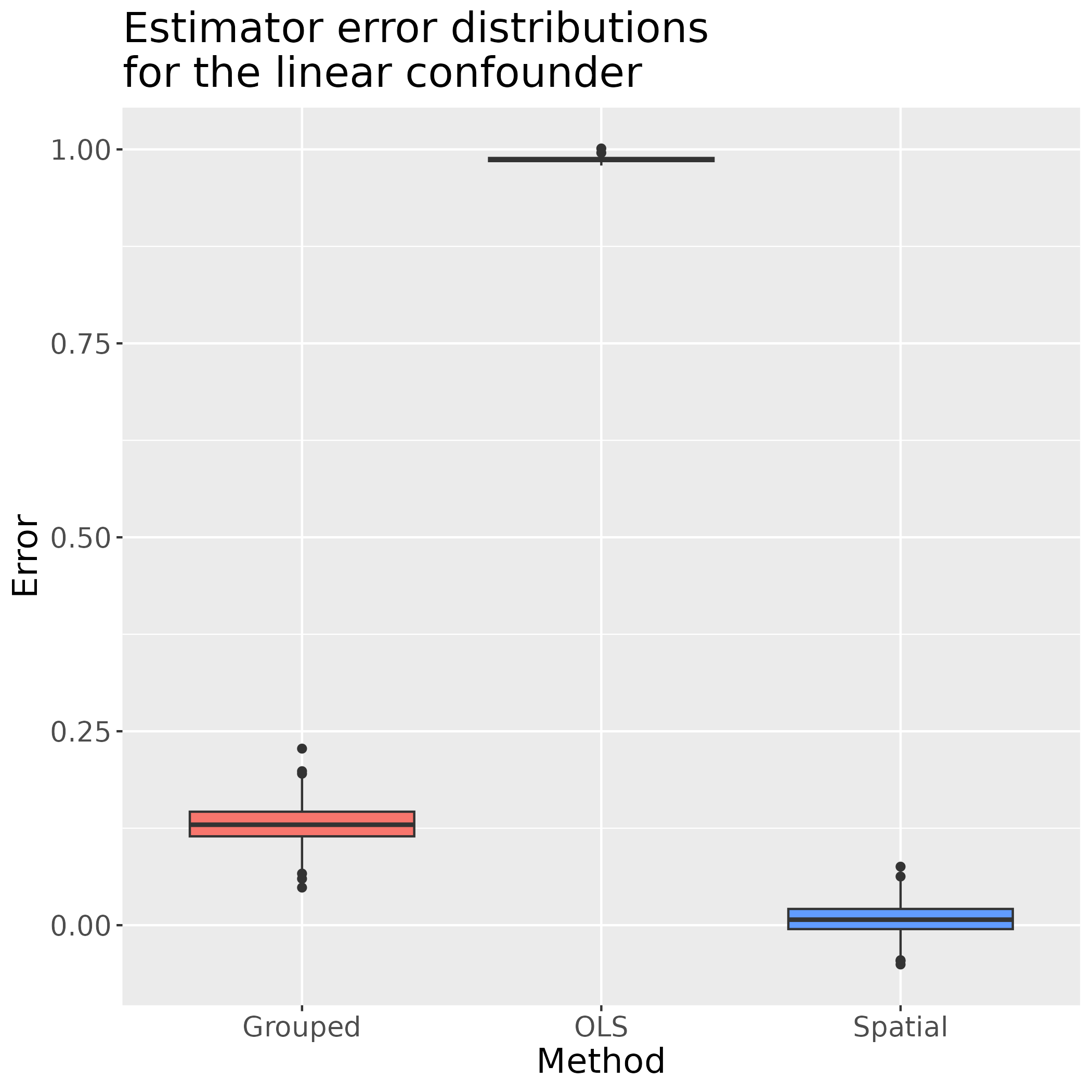}
    \caption{Estimator sampling distributions for the linear confounder. OLS is the unadjusted OLS estimator, grouped is the GLS estimator from the model (\ref{eq:glsdisc}) with independent location-specific random effects, Spatial is the GLS estimator from the model (\ref{eq:glscont}) with spatially smooth random effects modeled as a Gaussian process.}
    \label{fig:avsg}
\end{figure}

\begin{table}[ht]
    \centering
    \begin{tabular}{lrr}
        \hline
        & Grouped & Spatial \\
        \hline
        GLS Coverage & 0\% & 91\% \\
        \hline
    \end{tabular}
    \caption{GLS 95\% confidence interval coverage probabilities for both estimators for the linear confounder}
    \label{table:gls_coverage}
\end{table}

As expected, from Figure \ref{fig:avsg}, we see the OLS estimator is clearly inadequate. Although the grouped GLS estimator reduces bias, the bias is still significant (more than $10\%$ of the true effect size of $1$), and the estimator does not achieve nominal confidence interval coverage. On the other hand, although it is misspecified (does not model endogeneity explicitly), the spatial GLS estimator performs well with very little bias and nearly nominal coverage. The slight undercoverage is probably due to regularization bias, as discussed previously in Section \ref{sec:sim_bias}.

\bibliography{bib}

\end{document}